\documentclass[pra,showpacks, showkeys, reprint, amsart,amsmath,amssymb,aps]{revtex4-2}
\usepackage{graphicx,subcaption}
\usepackage{anyfontsize}
\usepackage{dcolumn}
\usepackage{bm}
\usepackage{epstopdf}
\usepackage{xcolor}
\usepackage{nccmath}
\usepackage{amsmath}
\usepackage{amsthm}
\usepackage{float}
\usepackage{multirow}
\usepackage{appendix}
\newtheorem{theorem}{Theorem}

\usepackage{balance}
\usepackage{silence}
\usepackage{mathtools}
\newlength{\arrow}
\newcommand{\bra}[1]{\langle#1|} 
\newcommand{\ket}[1]{|#1\rangle} 
\begin{document}
	
	\preprint{APS/123-QED}
	
	\title{Construction of a Non-Linear Entanglement Witness Operator in Arbitrary Dimension Using a Given Linear Witness Operator }
	
	\author{Sonia}
	\email[]{sonia\_24phdam02@dtu.ac.in}
	\author{Satyabrata Adhikari}
	\email[]{satyabrata@dtu.ac.in}
	\affiliation{Department of Applied Mathematics,\\ Delhi Technological University, Delhi-110042, India}
\begin{abstract}
Entanglement detection is one of the important problems in quantum information theory. To deal with this problem, many entanglement detection criteria have been proposed. Among the proposed criteria, the detection of entanglement through witness operator (also known as linear entanglement witness (LEW) operator) may be considered as the most practical. Although the witness operator approach to detect entanglement is experimentally friendly, the construction of these operators is not a very simple task. Even if we are able to construct a LEW operator, our problem is not solved as it may either detect a few entangled states or not a single entangled state from a given family of entangled states. Thus, we need a constructive approach in order to tackle this type of problem. In this work, we provide a few constructions of the non-linear entanglement witnesses (NLEW) for $d_1\otimes d_2$ dimensional system from any linear entanglement witness (LEW) operator. The advantage of these constructions is that, if a LEW is unable to detect any particular entangled state described by the density operator $\rho^{ent}$ then our construction of NLEW may detect the same entangled state $\rho^{ent}$. Further, we have constructed NLEW operator that may detect not only a class of bipartite negative partial transpose entangled state (NPTES), but also positive partial transpose entangled state (PPTES). Moreover, we have shown that the constructed NLEW operators may be decomposed in terms of the tensor product of local observables and hence may be realizable in an experiment.  
\end{abstract}
	
	\keywords{}
	\maketitle{}
	
	\section{Introduction } 
Quantum entanglement \cite{EPR_1935,Schrodinger_1935} is defined as an unique phenomenon of quantum mechanics in which the two or more system of particles situated far apart from each other, are interconnected in such a way that if measurement is performed on the state of any one subsystem then it will immediately effect the state of other subsystem. This property is also known as an emergent property \cite{Caltech} and 
can be considered as future of quantum technologies. It is widely used as quantum resource in many different quantum communication tasks like quantum teleportation \cite{Wootters_1993,X-Mhu_2023}, remote state preparation \cite{Bennett_2001,pati_2000}, superdense coding \cite{bennett_1992,Nielsen_2000,harrow_2004}, quantum key distribution schemes \cite{ekert_1991}, quantum error correction \cite{brun_2006, bravyi_2025}, quantum algorithms \cite{ekert_1998} and many others \cite{Horodecki_2006}. 
 In order to perform these quantum communication tasks, it is necessary to generate entanglement between two or more subsystems. But due to the unavoidable presence of noise and imperfections during the process of generation of entanglement, it is quite difficult to identify the presence of an entanglement in the prepared composite system comprised of two or more subsystems \cite{Adhikari_2018}. Thus the detection of the presence of entanglement in an unknown given quantum state is one of the important problems to consider.\\
In the literature, there are many criterion that may detect the entanglement. The first criteria for the detection of entanglement was introduced by Peres and Horodecki \cite{peres_1996, Horodecki_1996}, which is based on the partial transposition operation. The Peres-Horodecki entanglement detection criterion states that if all the eigenvalues of the partial transposed matrix are positive then the corresponding bipartite quantum states are separable. On the other hand, if there exist at least one negative eigenvalue of the partial transposed matrix then the corresponding state is detected as an entangled state. Although this criterion is necessary and sufficient for $2\otimes 2$ and $ 2\otimes 3$ dimensional systems but it provides only necessary condition for the separability of a bipartite state in a higher dimensional system. The Peres-Horodecki criterion is also known as Positive Partial Transposition (PPT) criterion. We can also recall here that the partial transposition map is not a 2-positive map and thus it does not represent a completely positive map. Hence, it may not be implementable in an experiment although partial transposition map plays a vital role in the detection of entanglement. Therefore, this may be a major drawback of the partial transposition map in the sense of practical implementation. On the contrary, it is well known that positive but not completely positive (PNCP) maps may be considered as the strong detector of entanglement \cite{rohira_2021,scala_2024}. Range criterion \cite{Horodecki_1997} is another entanglement detection criterion that has additional power of detecting the positive partial transpose entangled states (also known as bound entangled states) but on the other hand Peres-Horodecki criterion is lacking this power of the detection of the bound entangled states. Although the Range criterion is powerful but it is very difficult to apply this criterion for higher dimensional systems as it requires spanning set of product states, which is not algorithmically efficient for large systems \cite{R Kumar_2026}. There also exist other entanglement detection criterion such as Reduction criterion \cite{Horodecki_1999}, Majorization criterion \cite{Nielsen_2001} and Realignment criterion \cite{Rudolph_2003,Chen_2002}.\\
All the above mentioned entanglement detection criterion are very powerful but their implementation in the laboratory are very less probable. Therefore, it will be very legitimate to have some entanglement detection criterion that may work in a realistic scenario. Thus, our search for the hermitian operator that may detect entangled state and which may be realizable in an experiment, is fulfilled by the entanglement witness operator. In many experiments, we find that  entanglement witnesses (EWs) are used for entanglement detection \cite{Terhal_2001,Horodecki_2009,Vogel_2009}. This task can be achieved by defining two inequalities on the expectation value of EW operator: first inequality holds for all separables states and the second inequality hold for at least one entangled state. In the language of mathematical inequality, a hermitian operator $W$ is said to be an EW operator if the following two inequalities holds:
\begin{eqnarray}
&&(i) Tr(W\rho_{sep})\geq 0,~\forall~~ \text{separable state}~ \rho_{sep}\nonumber\\&&
(ii) Tr(W\rho_{ent})< 0, \text{for at least one entangled state}~ \rho_{ent}\nonumber\\
\label{witdef}
\end{eqnarray}
There always exist atleast one EW operator that may detect an entangled state \cite{Horodecki_1996} but the problem lies in the construction of it. Therefore, the construction of an EW operator for a given entangled state is a challenging task. Also, determining EWs for all entangled states is known to be NP hard problem \cite{Gurvits_2004,Doherty_2004,Hou_2010}. Threafter, we will use the term linear entanglement witness (LEW) to represent a conventional EW operators.\\
Despite of going through many problems, a significant progress has been achieved in constructing LEWs \cite{Vogel_2009,Doherty_2004,Hou_2010,Chruscinski_2009}.
Although LEWs may play a vital role in the detection of entanglement but the performance of them is not satisfactory. They may either fail to detect many useful entangled states or may detect a few entangled states from a given family of states. Therefore, it is natural to ask for the existence of a non-linear entanglement witness (NLEW) operator for entanglement detection that may outperform some existing LEW? The answer of this question is in affirmative. It was found in a few works that the introduction of a non-linear term in entanglement detection criteria works better than LEW \cite{Duan_2000,Uffink_2002,Hoffmann_2003,Guhne_2004,Guhne_2006}. Therefore,
the problem of entanglement detection may be tackled through NLEW, which may throw some new insights on entanglement detection problem. Generally, we expect that NLEW detects larger class of entangled state when compared to LEW. This is due to the geometry of the expectation value of the witness operator with respect to all quantum state in a given dimension. The former operator i.e. NLEW represent a curved surface whereas the latter operator i.e. LEW representing a hyperplane \cite{Karimipour_2025}. Therefore, it would be interesting to find the analytical structure of NLEW, which may efficiently detect a family of entangled states in comparison to LEW.\\
We are now in a position to list the following points that strongly motivates us to do this work:\\
(M1) The construction of either any LEW or any NLEW, is not a trivial task instead of having a prior partial information about the state under investigation.\\
(M2) Even if we are able to construct a LEW denoted by $W_{L}$, it may happen that the constructed LEW fail to detect a full family or some member of the family of a bipartite mixed entangled state described by the density operator $\rho_{ent}$. That is, the linear witness operator $W_{L}$ satisfies the inequality  $Tr(W_L\rho_{ent})\geq 0$ for some entangled state $\rho_{ent}$ belonging to family of states. If this is the case then we have to search for a new linear entanglement witness operator $W_{L}^{new}$ that may detect more entangled state in the family of  $\rho_{ent}$ than $W_{L}$ but the construction of $W_{L}^{new}$ is not a trivial task.\\
Therefore, by seeing the above points, we would like to develop a systematic procedure for the construction of a NLEW using a given LEW, which may detect more entangled states than the given LEW operator.\\   
The work can be divided into the following sections. In section II, we have noted down the previous results that would be needed in the subsequent sections. Section III discusses the general form of NLEW in terms of a given LEW. In section IV, we have studied the construction of the NLEW operators to detect bipartite NPTES lying in a $d_1\otimes d_2$ dimensional  system. Also, we have discussed the efficiency of the constructed NLEW over LEW using a few examples. 
In section V, firstly, we have constructed LEW operator from the well known computable cross-norm realignment (CCNR) and De vicente (DV) entanglement detection criterion and then using it, the NLEW operators have been constructed. In section VI, we have provided a new construction of NLEW operators for the detection of PPTES. In section VII, we analytically discuss the exprimental realization of the constructed NLEW. Section VIII concludes the whole work.
\section{Some Definitions and Established Results}
\textbf{Result-1 \cite{Lasser_1995}:} If the operators $A$ and $B$ are hermitian then the following inequality holds 
\begin{eqnarray}
\lambda_{min}(A)Tr(B)\leq Tr(AB)\leq \lambda_{max}(A)Tr(B)		
\label{result-1}
\end{eqnarray}
\textbf{Result-2 \cite{Lin_2016,Zhang_2019}:} For a density operator $\rho_{AB}$ in $d_{1}\otimes d_{2}$ dimensional system, if $\rho_{AB}^{T_{B}}\geq0$ i.e. if the state is a PPT state, then the following inequalities hold
\begin{equation}
det(I_{d_{2}}+Tr_{A}(\rho_{AB}))\leq det(I_{d_{1}d_{2}}+\rho_{AB})
\label{det 1st identity}
\end{equation}
\begin{equation}
det(I_{d_{1}}+Tr_{B}(\rho_{AB}))\leq det(I_{d_{1}d_{2}}+\rho_{AB})
\label{det 2nd identity}
\end{equation}
Where $Tr_{A}(\rho_{AB})$ and $Tr_{B}(\rho_{AB})$ symbolize the partial trace of the state $\rho_{AB}$  with respect to the subsystems $A$ and $B$ respectively. $I_{d}$ represents the identity matrix of order $d\times d$ while $det$ stands for the matrix determinant operation.\\
\textbf{Result-3 \cite{Lin_2016}:} If a density operator $\rho_{AB}$ in a $d_{1}\otimes d_{2}$ dimensional system is PPT then both the operators given below are positive semidefinite i.e.
	\begin{equation}
		Tr_{B}(\rho_{AB})\otimes I_{d_{2}}-\rho_{AB} \geq 0
		\label{1st operarator}
	\end{equation} 
	\begin{equation}
		I_{d_{1}}\otimes Tr_{A}(\rho_{AB})-\rho_{AB} \geq 0
		\label{2nd operator}
	\end{equation}
	where $Tr_{A}(\rho_{AB})$
	and $Tr_{B}(\rho_{AB})$ are defined as the partial trace of state $\rho_{AB}$  with respect to the subsystems $A$ and $B$ and $I_{d}$ denote the identity matrix of order $d\times d$.\\
\textbf{Result-4 \cite{Lin_2016}:} If $\rho_{AB}\in H_{d_{1}\otimes d_{2}}$ is positive semidefinite then 
\begin{equation}
\Big(\frac{det(Tr_A({\rho_{AB}}))}{d_1}\Big)^{d_1}\geq det(\rho_{AB})
\label{det inequality from lin}
\end{equation}
where $H_{d_{1}\otimes d_{2}}$ denote the Hilbert space of dimension $d_{1}d_{2}$.\\
\textbf{Result-5 \cite{Meenakshi_1999}:} If $A$ and $B$ represents two positive semidefinite matrices then eigenvalues of $AB$ are non-negative i.e\\
\begin{equation}
AB\geq0
\label{ev identity}
\end{equation}
\textbf{Result-6 \cite{Zhan_2002}:} If $M_n$ denote the set of all  $n\times n$ matrices and $A\in M_n$ then we have the following inequality
\begin{equation}
|Tr(A)|\leq\sum_{i=1}^{n}\sigma_{i}(A)=||A||_1
\label{traceA less norm A}
\end{equation}
where $\sigma_{i}(A)$ symbolizes the $i^{th}$ singular value of matrix $A$ and $||A||_1$ denote the trace norm of the matrix $A$.

\section{Functional form of non-linear witness operator in terms of a given linear witness operator: A Brief Discussion}
\noindent Many linear witness operators have already been existed in the literature but a few non-linear witness operator have been constructed for the detection of an entangled state. Thus, in this section, we discuss how one can think of the construction of a non-linear witness operator from a given linear witness operator.\\ 
To start with, let us consider a non-linear operator $W_{NL}$, which is a function of a given linear operator $W_{L}$. Mathematically, it can be expressed as 
\begin{eqnarray}
	W_{NL}=f(W_{L})
\end{eqnarray}
where $f(W_{L})$ is any non-linear function of $W_{L}$. Specifically, we can consider $f(W_{L})$ as the polynomial function of $W_{L}$. Further, if the linear operator $W_{L}$ satisfies (\ref{witdef}) then we say that $W_{L}$ is a linear entanglement witness operator that may detect an entangled state lying in $d_{1}\otimes d_{2}$ dimensional system. Additionally, if $f(W_{L})$ also satisfies (\ref{witdef}) for a given $d_{1} \otimes d_{2}$ dimensional LEW operator $W_{L}$ then we call $f(W_{L})$, a non-linear entanglement witness operator. Now, our task is to find such a polynomial functional form of the non-linear entanglement witness operator $f(W_{L})$ that may detect more entangled states than the LEW operator $W_{L}$.\\
Let us begin our search for the polynomial functional form of $f(W_{L})$ by using trial and error method. Consider the functional form of $f(W_{L})$ for $d_{1}\otimes d_{2}$ dimensional system as
\begin{align}
f(W_{L})=W_{L}^2+W_{L}
\label{nonlinear1}
\end{align}  
where $W_{L}$ denotes LEW operator in the same dimension.
Clearly, $f({W_{L}})$ given in (\ref{nonlinear1}) represent a non-linear function of $W_{L}$ but it is not yet known that whether $f(W_{L})$ also satisfies the properties of the witness operator given in (\ref{witdef}). Therefore, to investigate the entanglement detection property of $f(W_{L})$, we consider an arbitrary bipartite separable state  described by the density operator $\rho_{sep}$. Therefore, for  any $d_{1} \otimes d_{2}$ dimensional separable state $\rho_{sep}$, we have  
\begin{eqnarray}
Tr(	f(W_{L})\rho_{sep})&=&Tr(W_{L}^2\rho_{sep})+Tr(W_{L}\rho_{sep})\nonumber\\ &\geq& Tr(W_{L}^2\rho_{sep})\nonumber\\ &\geq& \lambda_{min}[{W_{L}^{2}}]\nonumber\\ &\geq& 0
\label{cond1}
\end{eqnarray}
In the second line, we have used the fact that $W_{L}$ is a LEW operator and thus $Tr(W_{L}\rho_{sep}) \geq 0$ and the third line uses the Result-1. Now, it remains to show that for a certain given $W_{L}$, there exist an entangled state $\rho_{ent}$ for which $Tr(f(W_{L})\rho_{ent})< 0$. In particular, if we consider a LEW operator $W_{L}^{(\psi^-)}$ acting on $2\otimes 2$ dimensional system then it can be expressed as \cite{Guhne_2009}
\begin{align}
	W_{L}^{(\psi^-)}=(|\psi^{-}\rangle_{AB}\langle \psi^{-}|)^{T_{B}}=\frac{1}{2}\begin{pmatrix}
		0&0&0&-1\\
		0&1&0&0\\
		0&0&1&0\\
		-1&0&0&0
	\end{pmatrix}
\end{align}
where $|\psi^{-}\rangle_{AB}=\frac{1}{\sqrt(2)}(|01\rangle_{AB}-|10\rangle_{AB})$ and $T_{B}$ denote the partial transposition with respect to the subsystem $B$. Moreover, it can be shown that the witness operator $W_{L}^{(\psi^-)}$ detected the entangled state described by the density operator $\rho_{ent}$, which is given by
\begin{eqnarray}
	\rho_{ent}=\begin{pmatrix}
		\frac{13}{30}&0&0&\frac{11}{30}\\
		0&\frac{1}{15}&0&0\\
		0&0&\frac{1}{15}&0\\
		\frac{11}{30}&0&0&\frac{13}{30}
	\end{pmatrix}
	\label{rho_entcons.}
\end{eqnarray}  
Therefore, for a given $W_{L}^{(\psi^-)}$, an entangled state $\rho_{ent}$ and for a  non-linear operator given in (\ref{nonlinear1}), we find that
\begin{eqnarray} 
Tr[f(W_{L}^{(\psi^-)})\rho_{ent}]<0
\end{eqnarray}
Thus, it shows that $f(W_{L}^{(\psi^{-})})$ is a well defined non-linear entanglement witness operator for a given linear entanglement witness operator $W_{L}^{(\psi^{-})}$.\\
However, we find that neither $W_{L}^{(\psi^{-})}$ nor $f(W_{L}^{(\psi^{-})})$ detect any member of a large class of a two-parameter family of $2\otimes2$ NPTES represented by the density matrix $\rho_{s,t}$, which is given by \cite{Rudolph_2003}.
\begin{align}
	\rho_{s,t}=\begin{pmatrix}
		\frac{5}{8}&0&0&\frac{t}{2}\\
		0&0&0&0\\
		0&0&\frac{1}{2}(s-\frac{1}{4})&0\\
		\frac{t}{2}&0&0&\frac{1-s}{2}
	\end{pmatrix}
\end{align} 
where $s\in [0.2926,0.3]$ and $t\in [0.02,0.0213]$.\\
The above discussion motivates us to construct an efficient NLEW operator to detect $d_{1} \otimes d_{2}$ dimensional entangled quantum state. By efficient NLEW operator, we mean to say that  $W_{NL}$ must be constructed using a linear witness operator $W_{L}$ in such a way that $W_{NL}$ detect larger set of entangled states in comparison to $W_{L}$.\\
Therefore, in the running section, we have provided a general idea of constructing a NLEW operator from a given LEW operator to detect the $d_{1}\otimes d_{2}$ dimensional entangled quantum states.\\  

\section{Non-Linear Entanglement Witness Operator to detect NPTES in $d_{1}\otimes d_{2}$ dimensional system}  
In this section, we provide an explicit form of NLEW operator expressed in terms of a non-linear function of a given LEW operator to detect $d_{1}\otimes d_{2}$ dimensional NPTES. We will also discuss the efficiency of the constructed NLEW operator over the LEW operator, from which it has been constructed.\\
Let us consider a LEW operator $W_{L}$, an entangled state $\rho_{ent}$ and a separable state $\rho_{sep}$. Any non-linear function of $W_{L}$ denoted by $W_{NL}$ is said to be a NLEW operator if it satisfies the following conditions:
\begin{eqnarray}
	&&\textbf{C1.}~~ Tr(W_{NL}\rho_{sep})\geq 0,~\forall~  \rho_{sep}\nonumber\\&&
	\textbf{C2.}~~  Tr(W_{NL}\rho_{ent})< 0, \text{for at least one}~ \rho_{ent}\nonumber\\&&
	\label{nlwitdef}
\end{eqnarray}
\subsection{Construction of Non-Linear entanglement witness operator to detect NPTES}
\noindent Our task is now to construct a NLEW operator using a given LEW operator. The systematic way of construction is given in the theorem below:
\begin{theorem}
If $W_L$ is a linear entanglement witness operator  in $d_{1}\otimes d_{2} $ dimensional system and $\rho_{AB}$ denoting the state under probe, then the non-linear witness operator $W_{NL}^{(1)}$ may take the form as
\begin{equation}
	\begin{split}
		W_{NL}^{(1)}&=\frac{1}{d_{1}d_{2}}W_L^{2}+\frac{d_{1}d_{2}}{d_{1}d_{2}+1}W_{L}-\\
		&\big(det(I_{d_{2}}+Tr_{A}(\rho_{AB}))-det(I_{d_{1}d_{2}}+\rho_{AB})\big)I_{d_{1}d_{2}}
	\end{split}
	\label{WNL1}
\end{equation} 
where $Tr_{A}(\rho_{AB})$ is partial trace of the state $\rho_{AB}$  with respect to the subsystem $A$, $I_{d_{1}d_2}$ is the identity matrix of order $d_1\times d_2$ and $det (.)$ denote the determinant of the matrix $(.)$.
\end{theorem}
\begin{proof}
To prove $W_{NL}^{(1)}$ given in (\ref{WNL1}), a witness operator, we need to show the two conditions given in (\ref{nlwitdef}).\\
(i) Let $\rho_{AB}^{sep}$ be any arbitrary separable state in $d_1\otimes d_2$ dimensional system. Then the expectation value of the operator $W_{NL}^{(1)}$ with respect to $\rho_{AB}^{sep}$ is given by 
\begin{equation}
\begin{split}
Tr(W_{NL}^{(1)}\rho_{AB}^{sep})&=\frac{1}{d_{1}d_{2}}Tr(W_L^{2}\rho_{AB}^{sep})+\frac{d_{1}d_{2}}{d_{1}d_{2}+1}Tr(W_L\rho_{AB}^{sep})\\
&-(det(I_{d_{2}}+Tr_{A}(\rho_{AB}^{sep}))-det(I_{d_{1}d_{2}}+\rho_{AB}^{sep}))
\label{wnlpf}
\end{split}
\end{equation}
In the expression (\ref{wnlpf}), we can use the following facts:\\
\textbf{(F1)} Since $W_L$ is an entanglement witness operator so $Tr(W_{L}\rho_{AB}^{sep})\geq 0$.\\
\textbf{(F2)} Since the operators $W_{L}^{2}$ and $\rho_{AB}^{sep}$ are positive semidefinite operators so using the Result-5, we get $Tr(W_L^{2}\rho_{AB}^{sep}) \geq 0$.\\
\textbf{(F3)} By using the first part of Result-2 given in (\ref{det 1st identity}), we have
\begin{equation}
	\begin{split}
		k=det(I_{d_{2}}+Tr_{A}(\rho_{AB}^{sep}))-det(I_{d_{1}d_{2}}+\rho_{AB}^{sep})\leq 0
	\end{split}
	\label{det k constant}
\end{equation}
Using the facts \textbf{(F1)}, \textbf{(F2)}, \textbf{(F3)}  on any separable state $\rho_{AB}^{sep}$, the expression (\ref{wnlpf}) reduces to the following inequality
\begin{equation}
Tr(W_{NL}^{(1)}\rho_{AB}^{sep})\geq0,~~~ \forall \rho_{AB}^{sep}
\label{cond 1}
\end{equation}
(ii) The next task is to show that there exist an entangled state that may be detected by $W_{NL}^{(1)}$. To achieve this, let us consider an entangled state  $\ket{\phi^{+}}_{AB}=\frac{1}{\sqrt{2}}(\ket{00}+\ket{11})$ and a LEW operator $W_{L}^{p}$ \cite{Milne_2015}, which is given by
\begin{align}
W_{L}^{p}=\frac{1}{2}\begin{pmatrix}
p&0&0&0\\
0&1-p&1&0\\
0&1&1-p&0\\
0&0&0&p
\end{pmatrix}, \text{where}~~0<p\leq1
\label{WLP WO}
\end{align}
It may be easily verified that $W_{L}^{p}$ is a linear witness operator and it does not detect the state $\ket{\phi^{+}}_{AB}$ for any value of the parameter $0<p\leq 1$.\\
If $\rho_{\ket{\phi^{+}}_{AB}\bra{\phi^{+}}}$ denote the density operator corresponding to the state $\ket{\phi^{+}}_{AB}$ then the expectation value of $W_{NL}^{(1)}$ with respect to the state $\rho_{\ket{\phi^{+}}_{AB}\bra{\phi^{+}}}$ is given by
\begin{align}
\begin{split}
&Tr(W_{NL}^{(1)}\rho_{\ket{\phi^{+}}_{AB}\bra{\phi^{+}}})=\frac{1}{4}Tr((W_{L}^{p})^{2}\rho_{\ket{\phi^{+}}_{AB}\bra{\phi^{+}}})+\\&
\frac{4}{5}Tr(W_{L}^{p}\rho_{\ket{\phi^{+}}_{AB}\bra{\phi^{+}}})-[det(I_{2}+Tr_{A}(\rho_{\ket{\phi^{+}}_{AB}\bra{\phi^{+}}}))\\&-det(I_{4}+\rho_{\ket{\phi^{+}}_{AB}\bra{\phi^{+}}})]
\end{split}
\label{wnl1ent}
\end{align}
A simple calculation gives the following values  
\begin{equation}
\begin{split}
Tr((W_{L}^{p})^{2}\rho_{\ket{\phi^{+}}_{AB}\bra{\phi^{+}}})&=\frac{p^2}{4}\\
Tr(W_{L}^{p}\rho_{\ket{\phi^{+}}_{AB}\bra{\phi^{+}}})&=\frac{p}{2}\\
det(I_4+\rho_{\ket{\phi^{+}}_{AB}\bra{\phi^{+}}})&=2\\
det(I_2+Tr_{A}(\rho_{\ket{\phi^{+}}_{AB}\bra{\phi^{+}}}))&=\frac{9}{4}
\label{expressval}
\end{split}
\end{equation}
Inserting the values given in (\ref{expressval}), the equation  (\ref{wnl1ent}) reduces to
\begin{align}
\begin{split}
Tr(W_{NL}^{(1)}\rho_{\ket{\phi^{+}}\bra{\phi^{+}}})&=\frac{1}{80}(-20+p(32+5p))\\
&<0,~~~0<p\leq0.573
\end{split}
\label{cond2}
\end{align}\\
Combining the conditions given in (\ref{cond 1}), (\ref{cond2}) and using the fact that $W_{NL}^{(1)}$, a non-linear function of $W_{L}$, we can now infer that the operator $W_{NL}^{(1)}$ given in (\ref{WNL1}) is a valid NLEW operator, which may detect a bipartite $d_{1}\otimes d_{2}$ dimensional NPTES.
\end{proof}
\begin{theorem}
If $W_L$ is a linear entanglement witness operator  in $d_1\otimes d_2$ dimensional system, the density operator $\rho_{AB}$ represents a state under investigation and $Tr_{A}(\rho_{AB})$ denotes the partial trace of state $\rho_{AB}$ with respect to the subsystem $A$ then the non-linear operator 
\begin{equation}
\begin{split}
W_{NL}^{(2)}&=W_L^{2}+(d_1d_2)^2W_L-\frac{1}{d_{2}}\big(det(I_{d_{2}}+Tr_{A}(\rho_{AB}))-\\
&det(I_{d_{1}d_{2}}+\rho_{AB})\big)I_{d_{1}d_{2}}+d_{2}(I_{d_{1}}\otimes Tr_{A}(\rho_{AB})-\rho_{AB})
\end{split}
\label{NEW W0 3*3}
\end{equation}
represent a NLEW operator, where $(I_{d_{1}}\otimes Tr_{A}(\rho_{AB})-\rho_{AB})$ denote a positive semidefinite operator.
\end{theorem}
\begin{proof}
It is evident from the construction that  $W_{NL}^{(2)}$ represent a non-linear operator. The task remain to show whether $W_{NL}^{(2)}$ also satisfies the conditions $C1$ and $C2$.\\
\textbf{(i) Condition C1:} If $\rho_{AB}^{sep}$ denote an arbitrary bipartite separable state in $d_1\otimes d_2$ dimensional  system then the expectation value of the non-linear operator $W_{NL}^{(2)}$ can be calculated as 
\begin{equation}
\begin{split}
Tr(W_{NL}^{(2)}\rho_{AB}^{sep})&=Tr(W_L^{2}\rho_{AB}^{sep})+(d_1d_2)^2Tr(W_L\rho_{AB}^{sep})-\\
&\frac{1}{d_{2}}\big(det(I_{d_{2}}+Tr_{A}(\rho_{AB}^{sep}))-\\
&det(I_{d_{1}d_{2}}+\rho_{AB}^{sep})\big)+\\
&d_{2}Tr((I_{d_{1}}\otimes Tr_{A}(\rho_{AB}^{sep})-\rho_{AB}^{sep})\rho_{AB}^{sep})
\end{split}
\label{NEW W0 3*3 1}
\end{equation}
From (\ref{2nd operator}), it is known that $ I_{d_{1}}\otimes Tr_{A}(\rho_{AB}^{sep})-\rho_{AB}^{sep}$ is a positive semidefinite operator and a separable state described by the density matrix $\rho_{AB}^{sep}$ also represents a positive semidefinite operator, therefore by using (\ref{ev identity}), we can say that $ (I_{d_{1}}\otimes Tr_{A}(\rho_{AB}^{sep})-\rho_{AB}^{sep})\rho_{AB}^{sep}\geq 0$ i.e. the eigenvalues of the operator $(I_{d_{1}}\otimes Tr_{A}(\rho_{AB}^{sep})-\rho_{AB}^{sep})\rho_{AB}^{sep}$ are non- negative. Since $Tr(W_{L}\rho_{AB}^{sep})\geq0$, $Tr(W_{L}^{(2)}\rho_{AB}^{sep})\geq 0$, $det(I_{d_{2}}+Tr_{A}(\rho_{AB}^{sep}))-det(I_{d_{1}d_{2}}+\rho_{AB}^{sep})<0$ for PPT states and $ Tr((I_{d_{1}}\otimes Tr_{A}(\rho_{AB}^{sep})-\rho_{AB}^{sep})\rho_{sep})\geq 0$, so we have
\begin{equation}
\begin{split}
Tr(W_{NL}^{(2)}\rho_{AB}^{sep})&\geq 0,~~~~\forall \rho_{AB}^{sep}
\label{wnl2sep}
\end{split}
\end{equation}
\textbf{(i) Condition C2:} To show that there exist atleast one entangled state, which is detected by $W_{NL}^{(2)}$, we recall the witness operator $W_L^{p}$ defined in (\ref{WLP WO}) and consider a two qubit state defined as
	\begin{equation}
		\rho_{a}=\begin{pmatrix}
		\frac{a}{2}&0&0&-\frac{a}{2}\\
		0&1-a&0&0\\
		0&0&0&0\\
		-\frac{a}{2}&0&0&\frac{a}{2}
	\end{pmatrix}, 0\leq a \leq 1
	\label{rho a state}
	\end{equation}
It may be noted that the state $\rho_{a}$ is entangled when $a\in (0,1]$ and can be verified that the witness operator $W_L^{p}$ can not detect any state from the class of entangled state described by the density operator $\rho_{a}$. In order to find the expectation value of $W_{NL}^{(2)}$ with respect to the state $\rho_{a}$, we require the values of some quantities which are given below
	\begin{equation}
		\begin{split}
			Tr(W_{L}^{p}\rho_{a})&=\frac{1}{2}(1-a+p(-1+2a))\\
			Tr((W_{L}^{p})^{2}\rho_{a})&=\frac{1}{4}(2+p^2+2p(-1+a)-2a)\\
			det(I_4+\rho_{a})&=2+a-a^2\\
			det(I_2+Tr_{A}(\rho_{a}))&=\frac{1}{4}(8+2a-a^2)
		\end{split}
		\label{rhoa quantities}
	\end{equation}
	Using (\ref{rhoa quantities}), the expectation value of $W_{NL}^{(2)}$ with respect to the state $\rho_{a}$ can be calculated as
	\begin{equation}
		\begin{split}
		Tr(W_{NL}^{(2)}\rho_a)&=\frac{1}{8}(68+2p^2-50a-27a^2+4p(-17+33a))\\
		&<0, ~~0.917\leq a\leq1,~~0<p\leq 0.01044
		\end{split}
	\end{equation}	
	Thus, the non-linear operator $W_{NL}^{(2)}$ satisfies all criterion of an entanglement witness operator. Hence, $W_{NL}^{(2)}$ is proved to be a legitimate NLEW operator.
\end{proof}
\begin{theorem}
If $W_{L}$ and $\rho_{AB}$ denotes LEW operator and a bipartite state in $d_1\otimes d_2$ dimensional system respectively then
\begin{equation}
\begin{split}
W_{NL}^{(3)}&=W_L^{2}-\frac{1}{(d_1d_2)^2}\big(det(I_{d_{2}}+Tr_{A}(\rho_{AB}))-\\
&det(I_{d_{1}d_{2}}+\rho_{AB})\big)I_{d_{1}d_{2}}+\\& (d_1d_{2})^{2}(I_{d_{1}}\otimes Tr_{A}(\rho_{AB})-\rho_{AB})
\end{split}
\label{NEW W0 2*4}
\end{equation}
is a NLEW operator, where the notations have their usual meaning.
\end{theorem} 
\begin{proof}
By construction, it can be easily proved that $Tr(W_{NL}^{(3)}\rho_{AB}^{sep})\geq 0,~~\forall \rho_{AB}^{sep}$. Therefore, the only thing remain to show that there exist at least one entangled state, which is detected by $W_{NL}^{(3)}$.\\ 
To achieve this aim, let us consider a $2\otimes3$ NPTES given by \cite{Chi_2003}:
\begin{equation}
	\rho_{2\otimes3}=\begin{pmatrix}
		0&0&0&0&0&0\\
		0&\frac{1}{2}&0&-\frac{1}{2}&0&0\\
		0&0&0&0&0&0\\
		0&-\frac{1}{2}&0&\frac{1}{2}&0&0\\
		0&0&0&0&0&0\\
		0&0&0&0&0&0
	\end{pmatrix}
\label{entstatewnl3}
\end{equation}
Further, consider a linear witness  which represents a state in subspace of six-dimensional vector space and is of the form as \cite{Sen_2023,Guhne_2009} 
\begin{equation}
	W_L^{\ket{\psi^-}}=\ket{\psi^{-}}\bra{\psi^{-}}^{T_{B}}
\end{equation}
where $\ket{\psi^{-}}=\frac{1}{\sqrt{2}}(\ket{01}-\ket{10})$.\\
In this case, we found that $Tr(W_L^{\ket{\psi^-}}\rho_{2\otimes3})>0$ and thus it may be concluded that the LEW operator does not detect the state given in (\ref{entstatewnl3}). However, if we construct a NLEW operator defined in (\ref{NEW W0 2*4}) by using $W_L^{\ket{\psi^-}}$, then we can have the following inequality
\begin{equation}
	Tr({W_{NL}^{(3)}}\rho_{2\otimes3})<0
\end{equation}
Thus, it proves that there exist an entangled state detected by (\ref{NEW W0 2*4}). Hence $W_{NL}^{(3)}$ is a NLEW operator.
\end{proof}
\subsection{Efficiency of the non-linear witness operator over the linear witness operator}
In this section, we analyze the efficiency of the constructed NLEW operators $W_{NL}^{(1)}, W_{NL}^{(2)}$ and $W_{NL}^{(3)}$. But since we are interested here to analyze the efficiency of the NLEW operator over the given LEW operator so we will consider those LEW operator that are weak in the context of the detection of entangled state. By the weak LEW operator, we mean that those LEW operator which are either incapable of detecting any entangled state or a few entangled states in a given family of the quantum state. To go into the depth of our analysis, we will consider a few examples of states in different dimensional system.
\subsubsection{Efficiency of $W_{NL}^{(1)}$ in $2\otimes 2$ dimensional system}
\noindent Let us consider a LEW operator expressed in matrix form as \cite{Milne_2015}:\\
\begin{align}
	W_{L}^{p}=\frac{1}{2}\begin{pmatrix}
		p&0&0&0\\
		0&1-p&1&0\\
		0&1&1-p&0\\
		0&0&0&p
	\end{pmatrix}, \text{where}~~0<p\leq1
\end{align}
\textbf{Example-1}. Let us consider the two-qubit \emph{isotropic} state defined by the density operator $\rho_{AB}^{\alpha}$ as follows \cite{Bertlmann_2009}
\begin{align}
	\rho_{AB}^{\alpha}=
	\begin{pmatrix}
		\frac{1+\alpha}{4}&0&0&\frac{\alpha}{2}\\
		0&  \frac{1-\alpha}{4}   &0&0\\
		0&0&\frac{1-\alpha}{4}&0\\
		\frac{\alpha}{2}&0&0&	\frac{1+\alpha}{4}
	\end{pmatrix},~~~ -\frac{1}{3}\leq\alpha\leq1
\end{align}
The \emph{isotropic} two-qubit state $\rho_{AB}^{\alpha}$ is separable when $\frac{-1}{3}\leq\alpha\leq\frac{1}{3}$ and NPTES when $\frac{1}{3}<\alpha\leq1$. We may note here that the LEW operator $W_{L}^{p}$ fail to detect this state $\rho_{AB}^{\alpha}$ as an entangled state i.e $Tr(W_{L}^{p}\rho_{AB}^{\alpha})\geq0$ for  $\frac{1}{3}<\alpha\leq 1 $ and $0<p\leq1$.\\
Now, we have
\begin{align}
	\begin{split}
		Tr((W_{L}^{p})^{2}\rho_{AB}^{\alpha})&=\frac{1}{4}((1-\alpha)(1-p)+p^2)\\
		Tr(W_{L}^{p}\rho_{AB}^{\alpha})&=\frac{1}{4}(1+(-1+2p)\alpha)\\
		det(I_4+\rho_{AB}^{\alpha})&=\frac{1}{256}(-5+\alpha)^2(25+10\alpha-3\alpha^2)\\
		det(I_2+Tr_{A}(\rho_{AB}^{\alpha}))&=\frac{9}{4}
	\end{split}
	\label{valuesalphastate}
\end{align}
The expectation value of $W_{NL}^{(1)}$ with respect to the state $\rho_{AB}^{\alpha}$ is given by
\begin{equation}
	\begin{split}
		Tr(W_{NL}^{(1)}\rho_{AB}^{\alpha})&=\frac{1}{d_{1}d_{2}}Tr((W_{L}^{p})^{2}\rho_{AB}^{\alpha})+\frac{d_{1}d_{2}}{d_{1}d_{2}+1}Tr(W_{L}^{p}\rho_{AB}^{\alpha})\\
		&-(det(I_{d_{2}}+Tr_{A}(\rho_{AB}^{\alpha}))-det(I_{d_{1}d_{2}}+\rho_{AB}^{\alpha}))
		\label{wnlex1}
	\end{split}
\end{equation}
Using the values given in (\ref{valuesalphastate}), the expectation value (\ref{wnlex1}) reduces to
\begin{align}
	\begin{split}
		Tr(W_{NL}^{(1)}\rho_{AB}^{\alpha})&=\frac{1}{1280}\Bigg(581+80p^2+16p(-5+37\alpha)-\\
		&\alpha\bigg(336+5\alpha(150+\alpha)\big(-40+3\alpha\big)\bigg)\Bigg)	\end{split}
\end{align}\\
We find that $Tr(W_{NL}^{p}\rho_{AB}^{\alpha})<0$ when $0.968\leq\alpha\leq 1$ and $0.5210\leq p\leq0.5213$. Therefore, for the parameter $p\in [0.5210,0.5213]$, there exist LEW operator $W_{L}^{p}$ which does not detect the state $\rho_{AB}^{\alpha}$ for any $\alpha \in (\frac{1}{3},1]$ but the NLEW operator $W_{NL}^{(1)}$ detect it when $\alpha \in [0.968,1]$.\\
\textbf{Example-2.} Consider a maximally entangled mixed state (MEMS) defined as \cite{Hiroshima_2000}:
\begin{align}
	\rho_{AB}^{(4)}=p_{1}\ket{\phi^{-}}\bra{\phi^{-}}+p_{2}\ket{01}\bra{01}+p_{3}\ket{\phi^{+}}\bra{\phi^{+}}+p_{4}\ket{10}\bra{10}
	\label{rho4}
\end{align}
where $\sum_{i=1}^{4}p_{i}=1$ and $\ket{\phi^{\pm}}=\frac{1}{\sqrt{2}}(\ket{00}\pm\ket{11})$.\\
If we now choose the state parameters as $p_{1}=q,~ p_{2}=0.00003,~ p_{3}=0.02,~p_{4}=0.97997-q,~\text{where}~~0\leq q\leq0.97997$, then the density matrix $\rho_{AB}^{(4)}$ reduces to
\begin{align}
	\rho_{AB}^{q}=	\begin{pmatrix}
		0.01+\frac{q}{2}&0&0&0.01-\frac{q}{2}\\
		0&0.00003&0&0\\
		0&0&0.97997-q&0\\
		0.01-\frac{q}{2}&0&0&0.01+\frac{q}{2}
	\end{pmatrix}
\end{align}
It can be easily shown that $\rho_{AB}^{q}$ represented an entangled state when $0\leq q\leq 0.97997$ and MEMS when $0.031\leq q\leq 0.97997$. Further, it can be shown that LEW operator $W_{L}^{p}$ does not detect any member of the family of entangled states described by the density operator $\rho_{AB}^{q}$. 
In this case, $Tr(W_{NL}^{(1)}\rho_{AB}^{q})$ can be calculated as
\begin{align}
	\begin{split}
		Tr(W_{NL}^{(1)}\rho_{AB}^{q})=&0.524201 + 0.0625p^2+p(-0.5065 + 0.925 q)\\
		&+(-0.0153706 - 0.770031 q)q
	\end{split}
\end{align}
Our finding suggest that $Tr(W_{NL}^{(1)}\rho_{AB}^{q})<0$ when $p \in [0.2450, 0.2475]$ and $q \in [0.875, 0.97997]$. Thus, NLEW operator has larger entanglement-detecting capability as compared to LEW operator.\\
Hence, we have shown here that there is a LEW operator $W_{L}^{p}$ which is unable to detect any member of the family of entangled state described by the density operator $\rho_{AB}^{\alpha}$ and $\rho_{AB}^{q}$ defined in  $2\otimes 2$ dimensional system but instead using $W_{L}^{p}$, we can construct a NLEW operator $W_{NL}^{(1)}$ that may detect some entangled state from the same considered family. This shows that it is possible to construct a more efficient NLEW operator such as $W_{NL}^{(1)}$ using the LEW operator $W_{L}^{p}$.\\
\subsubsection{Efficiency of $W_{NL}^{(1)}$ and $W_{NL}^{(2)}$ in $3\otimes 3$ dimensional system}
\noindent In this section, we study the efficiency of NLEW operator $W_{NL}^{(1)}$ and $W_{NL}^{(2)}$ over LEW operator in detecting the NPTES lying in $3\otimes 3$ dimensional system.\\
To do this, let us consider a linear witness as $\ket{\phi^+}\bra{\phi^+}^{T_B}$ where $\ket{\phi^+}=\frac{\ket{00}+\ket{11}}{\sqrt{2}}$ represent a state in a subspace of a nine-dimensional vector space, which can be expressed in matrix form as 
\begin{align}
	W_{L}^{3\otimes 3}=	\begin{pmatrix}
		\frac{1}{2}&0&0&0&0&0&0&0&0\\
		0&0&0&\frac{1}{2}&0&0&0&0&0\\
		0&0&0&0&0&0&0&0&0\\
		0&\frac{1}{2}&0&0&0&0&0&0&0\\
		0&0&0&0&\frac{1}{2}&0&0&0&0\\
		0&0&0&0&0&0&0&0&0\\
		0&0&0&0&0&0&0&0&0\\
		0&0&0&0&0&0&0&0&0\\
		0&0&0&0&0&0&0&0&0\\
	\end{pmatrix}
\end{align}
\textbf{Example-3.} Let us now consider a one-parameter family of isotropic states lying on a $3\otimes3$ dimensional system, i.e. a two-qutrit system \cite{Wang_2010}, which is defined as
\begin{equation}
	\rho_{AB}^{3\otimes3}(\gamma)=\frac{1-\gamma}{8}I_{9}+\frac{9\gamma-1}{8}\ket{\psi^+}\bra{\psi^+},~0\leq\gamma\leq1
	\label{isotropicstate}
\end{equation}
where $\ket{\psi^+}=\frac{1}{\sqrt{3}}( \ket{00}+\ket{11}+\ket{22})$ and $\gamma=\bra{\psi^+}\rho_{AB}^{3\otimes3}(\gamma)\ket{\psi^+}$. The state $\rho_{AB}^{3\otimes3}(\gamma)$ is separable when $0\leq \gamma\leq\frac{1}{3}$ and NPTES when $\frac{1}{3}<\gamma\leq 1 $. It can be seen that $Tr(W_{L}^{3\otimes 3}\rho_{AB}^{3\otimes3}(\gamma))\geq0$ for $\gamma \in [0,1]$ i.e. a LEW operator $W_{L}^{3\otimes3}$ fails to detect any member of the family described by the density matrix $\rho_{AB}^{3\otimes3}(\gamma)$.\\
Now, let us calculate a few quantities using LEW operator $W_{L}^{3\otimes3}$, which is given by
\begin{equation}
	\begin{split}
		Tr((W_{L}^{3\otimes3})^{2}\rho_{AB}^{3\otimes3}(\gamma))&=\frac{1}{48}(5+3\gamma)\\
		Tr((W_{L}^{3\otimes3})\rho_{AB}^{3\otimes3}(\gamma))&=\frac{1}{12}(1+3\gamma)\\
		det(I_9+\rho_{AB}^{3\otimes3}(\gamma))&=\frac{(-9+\gamma)^8(1+\gamma)}{16777216}\\
		det(I_3+Tr_{A}(\rho_{AB}^{3\otimes3}(\gamma)))&=\frac{64}{27}
	\end{split}
	\label{gamma state values 3-3}
\end{equation} 	
Using all the values given in (\ref{gamma state values 3-3}), we find that the expectation value of the operator $W_{NL}^{(1)}$ satisfies
\begin{equation}
	\begin{split}
		Tr((W_{NL}^{(1)})\rho_{AB}^{3\otimes3}(\gamma))<0,~~0.932\leq\gamma\leq1
	\end{split}
\end{equation} 
Thus, $W_{NL}^{(1)}$ detected a few member of the family of the state described by the density operator $\rho_{AB}^{3\otimes3}(\gamma)$, while linear witness operator fails to detect any member of the family.\\
\textbf{Example-4.}  In this example, we consider a LEW operator for $3\otimes 3$ dimensional system, which may be defined as \cite{Wudarski_2011}:
\begin{align}
	W_{L}^{c}=\frac{1}{33}\begin{pmatrix}
		1&0&0&0&-1&0&0&0&-1\\
		0&9&0&0&0&0&0&0&0\\
		0&0&1&0&0&0&0&0&0\\
		0&0&0&1&0&0&0&0&0\\
		-1&0&0&0&1&0&0&0&-1\\
		0&0&0&0&0&9&0&0&0\\
		0&0&0&0&0&0&9&0&0\\
		0&0&0&0&0&0&0&1&0\\
		-1&0&0&0&-1&0&0&0&1\\
	\end{pmatrix}
	\label{WLCWO}
\end{align}
Recalling a one-parameter family of isotropic state given in (\ref{isotropicstate}), we find that $W_{L}^{c}$ is detecting some entangled state in the family of isotropic state, which is evident from the following expression
\begin{equation}
	\begin{split}
		Tr(W_{L}^{c}\rho_{AB}^{3\otimes3}(\gamma))=&\frac{1}{132}(17-21\gamma)\\
		&<0,~0.81\leq \gamma\leq1
	\end{split}
	\label{WLqutritfstate}
\end{equation}
Using (\ref{isotropicstate}) and (\ref{WLCWO}), we calculate the following quantities
\begin{align}
	\begin{split}
		Tr(W_{L}^{c}\rho_{AB}^{3\otimes3}(\gamma))&=\frac{1}{132}(17-21\gamma)\\
		Tr((W_{L}^{c})^{2}\rho_{AB}^{3\otimes3}(\gamma))&=\frac{127-123\gamma}{4356}\\
		det(I_9+\rho_{AB}^{3\otimes3}(\gamma))&=\frac{(-9+\gamma)^8(1+\gamma)}{16777216}\\
		det(I_3+Tr_{A}(\rho_{AB}^{3\otimes3}(\gamma)))&=\frac{64}{27}
	\end{split}
	\label{qutrit f state}
\end{align}
Therefore, using the quantities given in (\ref{qutrit f state}), the expectation value of the NLEW operator $W_{NL}^{(2)}$ defined in (\ref{NEW W0 3*3}) with respect to the state $\rho_{AB}^{3\otimes3}(\gamma)$ can be calculated as 
\begin{equation}
	\begin{split}
		Tr(W_{NL}^{(2)}\rho_{AB}^{3\otimes3}(\gamma))=&\frac{1}{164433494016}\Bigg(1833616212611 - \\
		&1984641812181 \gamma - 631356623172 \gamma^2 + \\
		&	37810964268 \gamma^3 - 9302697558 \gamma^4 +\\ & 1367063082 \gamma^5-125962452 \gamma^6 +\\ & 174332 \gamma^7 - 231957 \gamma^8 + 3267 \gamma^9\Bigg)\\ & < 0,~~0.752\leq \gamma\leq1
	\end{split}
\end{equation}
Further, it can be easily seen that the non-linear entanglement witness operator $W_{NL}^{(2)}$ is detecting more entangled state in the family described by the density operator $\rho_{AB}^{3\otimes3}(\gamma)$ in comparison to the linear entanglement witness operator $W_{L}^{c}$.
\subsubsection{Efficiency of $W_{NL}^{(3)}$ in $2\otimes 4$ dimensional system}
Here, we will show the efficiency of the NLEW operator $W_{NL}^{(3)}$ defined for $2\otimes 4$ dimensional system over the LEW operator defined for the same dimensional system. For this, let us consider a Hermitian LEW operator for $2\otimes4$ dimensional system as $W_{L}^{\phi^{+}}=\ket{\phi^{+}}\bra{\phi^{+}}^{T_{B}}$\cite{Sen_2023,Guhne_2009} where $\ket{\phi^{+}}=\frac{1}{\sqrt{2}}(\ket{00}+\ket{11})$ represent a state in a subspace of a eight-dimensional vector space. It can be explicitly expressed in the matrix form as
\begin{align}
	W_{L}^{\phi^{+}}=\frac{1}{2}\begin{pmatrix}
		1&0&0&0&0&0&0&0\\
		0&0&0&0&1&0&0&0\\
		0&0&0&0&0&0&0&0\\
		0&0&0&0&0&0&0&0\\
		0&1&0&0&0&0&0&0\\
		0&0&0&0&0&1&0&0\\
		0&0&0&0&0&0&0&0\\
		0&0&0&0&0&0&0&0\\
	\end{pmatrix}
	\label{witphi}
\end{align}
\textbf{Example-5:} Let us now consider a family of NPTES in $H_2\otimes H_4$ dimensional system described by the density operator $\rho_{AB}^{(b)}$ as
\begin{align}
\rho_{AB}^{(b)}=\begin{pmatrix}
\frac{b}{6b+1}&0&0&0&0&0&0&\frac{b}{6b+1}\\
0&\frac{b}{6b+1}&0&0&0&0&\frac{b}{6b+1}&0\\
0&0&\frac{b}{6b+1}&0&0&\frac{b}{6b+1}&0&0\\
0&0&0&0&0&0&0&0\\
0&0&0&0&0&0&0&0\\
0&0&\frac{b}{6b+1}&0&0&\frac{b}{6b+1}&0&0\\
0&\frac{b}{6b+1}&0&0&0&0&\frac{b}{6b+1}&0\\
\frac{b}{6b+1}&0&0&0&0&0&0&\frac{b+1}{6b+1}\\
\end{pmatrix}
\label{rhob}
\end{align}
where $b\in[0,1]$.\\
It can easily shown that LEW operator $W_{L}^{\ket{\phi^{+}}}$ does not detect the entangled state $\rho_{AB}^{(b)}$. Using the LEW operator given in (\ref{witphi}) and the state $\rho_{AB}^{(b)}$ given in (\ref{rhob}), we can calculate the value of the following quantities as per the requirement given in (\ref{NEW W0 2*4}) as
\begin{align}
	\begin{split}
		Tr(W_{L}^{\ket{\phi^{+}}}\rho_{AB}^{(b)}))&=\frac{b}{1+6b}\\
		Tr((W_{L}^{\phi^{+}})^{2} \rho_{AB}^{(b)})&=\frac{3}{4(1+6b)}\\
		det(I_8+\rho_{AB}^{(b)})&=\frac{(1+8b)^2(2+21b+48b^2)}{(1+6b)^4}\\
		det(I_4+Tr_{A}(\rho_{AB}^{(b)}))&=\frac{(1+7b)(2+7b)(1+8b)^2}{(1+6b)^4}
	\end{split}
	\label{q1}
\end{align}
Using the calculated values given in (\ref{q1}), we can find the expectation value of $W_{NL}^{(3)}$ with respect to the state $\rho_{AB}^{(b)}$ as
\begin{equation}
	\begin{split}
		Tr(W_{NL}^{(3)}\rho_{AB}^{(b)})=&\frac{1}{64(1+6b)^4}\bigg(
		48 b - 7329 b^2 \\
		&- 93136 b^3 - 284608 b^4\bigg)
	\end{split}
\end{equation}
It can be verified that $Tr(W_{NL}^{(3)}\rho_{AB}^{(b)})<0$ for all $b\in (0,1]$. Therefore, in this example, we find that the $W_{NL}^{(3)}$ is more efficient than the LEW operator from which the NLEW operator has been constructed in the sense that the operator $W_{NL}^{(3)}$ detect more entangled state in comparison to the given LEW operator $W_{L}^{\phi^{+}}$.

\section{Construction of the NLEW operator from CCNR and DV separability criterion}
This section discusses the previously well-established separability criteria's such as Computable Cross-Norm Realignment (CCNR) or simply Realignment criterion and De vicente (DV) criterion. Thereafter, we will use the above discussed separability criterion to construct the LEW operator and henceforth, we will use the constructed LEW operator for further construction of NLEW operator. We will then show that the constructed NLEW operator detect the NPTES, which is not detected by the corresponding separability criterion. 
\subsection{CCNR Criteria}
CCNR criterion states that any separable state described by the density operator $\rho^{sep}$ satisfies the following inequality \cite{Rudolph_2003} 
\begin{equation}
||C||_1\leq1
\label{CCNR}
\end{equation}
where $C$ denotes the correlation matrix with the matrix elements generally defined as $C_{a,b}=Tr(\rho^{sep} G_a^{A}\otimes G_b^{B})$ for $a\geq 0$ and $b\geq 0$. $ G_a^{A}$ and  $ G_b^{B}$ represents the arbitrary orthonormal basis belonging to $\mathcal{B}(H_A)$ and $\mathcal{B}(H_B)$ respectively. However, If we choose $(G_a^{A})_{a\ne0}$ and $(G_b^{B})_{b\ne0}$ as an orthonormal traceless operators with  $G_0^{A}=\frac{I_{A}}{\sqrt{d_1}}$ and $G_0^{B}=\frac{I_{B}}{\sqrt{d_1}}$ then $\frac{1}{\sqrt{d_1}}G_a^{A}$$\in \mathcal{B}(H_A)$ and $\frac{1}{\sqrt{d_1}}G_b^{B}$$\in \mathcal{B}(H_B)$ forms a canonical basis. The correlation matrix in canonical bases may be denoted by $C^{can}$. It can be easily shown that $||C||_1=||C^{can}||_1$ \cite{Chruscinki_2020}.\\ 
CCNR criteria is a necessary but not a sufficient criterion. The contrapositive statement is therefore useful in detecting entangled state, which states that if any bipartite state violate the condition (\ref{CCNR}) then the state must be an entangled state. However, one can find that the necessary and sufficient condition hold for CCNR criterion for the following paticular classes of states \cite{Rudolph_2005}: (i) All pure states, (ii) Bell Diagonal states, (iii) $2\otimes 2$ dimensional Werner states and (iv) Bipartite isotropic states in arbitrary dimensions.  CCNR criterion can be used as a novel computable separability criteria to detect a few negative partial transpose entangled states as well as bound entangled states.
\subsubsection{Construction of LEW Using CCNR Criterion}
In this section, we have shown that a LEW operator can be constructed using CCNR separability criterion. The construction can be explained by the theorem given below:\\
\begin{theorem}
 If the state under investigation is described by the density operator $\rho_{AB}$ in $d\otimes d$ dimensional system then a linear witness operator denoted by $W_L^{CCNR}$ may be expressed in the following form 
\begin{equation}
	\begin{split}
		W_L^{CCNR}&=\Big(1-d^{2}(d+1)^2~k\Big)I_{d^{2}}-\frac{1}{\lambda_{\max}(\rho_{AB})}C^{can}
	\end{split}
	\label{WCCNR1}
\end{equation}
where $k=(det(I_{d}+Tr_{A}(\rho_{AB}))-det(I_{d^{2}}+\rho_{AB}))$ is defined in $(\ref{det k constant})$. $Tr_{A}(\rho_{AB})$ denotes the partial trace of the state $\rho_{AB}$ with respect to the subsystem $A$ and $I_{d^{2}}$ is the identity matrix of order $d^{2}\times d^{2}$. $C^{can}$ denote the correlation matrix of the state $\rho_{AB}$ in the canonical bases.\\
\end{theorem}
\begin{proof}
To prove the operator $W_L^{CCNR}$ as a witness operator, we need to fulfill two conditions given in (\ref{witdef}). To start with, let us consider an arbitrary bipartite separable state described by the density operator $\rho_{sep}$ in $d \otimes d$ dimensional system. Therefore, the expectation value of $W_L^{CCNR}$ with respect to the separable state $\rho_{sep}$ is given by 
\begin{equation}
\begin{split}
Tr(W_L^{CCNR}\rho_{sep})=&1-d^{2}(d+1)^2~k\\
&-\frac{1}{\lambda_{\max}(\rho_{sep})}Tr(C^{can}\rho_{sep})
\end{split}
\label{WCCNR2}
\end{equation}
Using the results given in (\ref{result-1}) and (\ref{traceA less norm A}) in (\ref{WCCNR2}), we get
\begin{equation}
	\begin{split}
		Tr(W_L^{CCNR}\rho_{sep})&\geq 1-||C^{can}||_1-d^{2}(d+1)^2~k
	\end{split}
	\label{WCCNR3}
\end{equation}
Since $k\leq0$ from (\ref{det k constant}) and $1-||C^{can}||_1\geq 0$ from CCNR separability criterion given in (\ref{CCNR}), we can conclude for any arbitrary separable state $\rho_{sep}$ that the following inequality holds
\begin{equation}
	Tr(W_L^{CCNR}\rho_{sep})\geq 0
	\label{wit1}
\end{equation}
To prove the second condition of the witness operator, let us consider a NPTES described by the density operator $\rho_{AB}^{(1)}$ in $2\otimes2$ dimensional system. The density operator $\rho_{AB}^{(1)}$ is given by   
\begin{equation}
	\begin{split}
		\rho_{AB}^{(1)}&=\begin{pmatrix}
			0.265822&0&0&0\\
			0& 0.367089 & -0.367089 &0\\
			0 & -0.367089 & 0.367089 & 0\\
			0&0&0&0
		\end{pmatrix}
	\end{split}
	\label{dstate}
\end{equation}
The expectation value of witness operator $W_L^{CCNR}$ with respect to the state $\rho_{AB}^{(1)}$ is given by
\begin{equation}
\begin{split}
Tr(W_L^{CCNR}\rho_{AB}^{(1)})&=-0.152209\\
&<0
\end{split}
\label{wit2}
\end{equation}
Equation (\ref{wit2}) clearly shows that there exist an entangled state $\rho_{AB}^{(1)}$ for which $Tr(W_L^{CCNR}\rho_{AB}^{(1)})<0$. Therefore, combining the equations (\ref{wit1}) and (\ref{wit2}), we can say that the operator $W_L^{CCNR}$ satisfies all the properties of the witness operator. Hence, it proves that $W_L^{CCNR}$ represent a linear entanglement witness operator.
\end{proof}
\subsubsection{Construction of NLEW Operator Using the LEW Operator $W_L^{CCNR}$}
\begin{theorem}
	If  $W_L^{CCNR}$ is a linear entanglement witness operator defined in (\ref{WCCNR1}) and $\lambda_{max}(W_L^{CCNR})$ denote the maximum eigenvalue of $W_L^{CCNR}$ then the non-linear entanglement witness operator may be constructed as  
\begin{equation}
	W_{NL}^{CCNR}=W_L^{CCNR}-\frac{1}{\lambda_{\max}(W_L^{CCNR})}(W_L^{CCNR})^2
	\label{WNLCCNR}
\end{equation}
\end{theorem}
\begin{proof}
(i) For any arbitrary bipartite separable state $\rho_{sep}$, the expectation value of $W_{NL}^{CCNR}$ is given by 
\begin{equation}
	\begin{split}
	Tr(	W_{NL}^{CCNR}\rho_{sep})=&Tr(W_L^{CCNR}\rho_{sep})-\\&
	\frac{1}{\lambda_{\max}(W_L^{CCNR})}Tr((W_L^{CCNR})^2\rho_{sep})
	\end{split}
	\label{WNL CCNR & DV 1}
\end{equation} 
Using Result-1 given in (\ref{result-1}), the equation (\ref{WNL CCNR & DV 1}) reduces to an inequality given by
\begin{equation}
	\begin{split}
		Tr(	W_{NL}^{CCNR}\rho_{sep})&\geq 0
	\end{split}
\label{witnl1}
\end{equation}
Therefore, the expectation value of the non-linear operator $W_{NL}^{CCNR}$ over all separable states is non-negative. Thus, it is only remaining to show that there exist an entangled state for which the expectation value of $W_{NL}^{CCNR}$ is negative. For this, recall the state $\rho_{AB}^{(1)}$ defined in (\ref{dstate}) and using this, let us calculate the following quantities as
\begin{equation}
	\begin{split}
		Tr(W_L^{CCNR}\rho_{AB}^{(1)})=&-0.152209\\
		Tr((W_L^{CCNR})^2\rho_{AB}^{(1)})=&0.304097\\
		\lambda_{\max}(W_L^{CCNR})=&0.161736\\
	\end{split}
	\label{dstatevalues}
\end{equation}
Using the value of the quantities given in (\ref{dstatevalues}), the expectation value of the non-linear operator $W_{NL}^{CCNR}$ with respect to the state $\rho_{AB}^{(1)}$ can be calculated as
\begin{equation}
	Tr(W_{NL}^{CCNR}\rho_{AB}^{(1)})=-2.0317<0
	\label{witnl2}
\end{equation}
Therefore, from (\ref{witnl1}) and (\ref{witnl2}), it is clear that the non-linear operator $W_{NL}^{CCNR}$ qualifies for the witness operator.
\end{proof}
\subsection{DV Criterion}
DV criterion states that if $C$ denote the correlation matrix of the bipartite separable state $\rho_{AB}^{sep}$ in $d \otimes d$ dimensional system then the following inequality is satisfied
\cite{Vicenta_2007}	
\begin{equation}
||C||_1\leq\frac{d(d-1)}{2}
\label{dv}
\end{equation}	
It may be observed that the CCNR and dV criteria will become equivalent for $2\otimes 2$ dimensional system. Further, if we use the correlation matrix $C^{can}$ in canonical basis and since $||C||_1=||C^{can}||_1$, then the above inequality (\ref{dv}) reduces to
\begin{equation}
||C^{can}||_1\leq\frac{d(d-1)}{2}
\label{dv1}
\end{equation}	
We find that the DV criterion usually weak in comparison to CCNR criterion when we are detecting entangled state in $d \otimes d$ dimensional system except $d=2$. However, for $d_{1} \otimes d_{2}$ dimensional system, the DV criterion may give better result than the CCNR criterion. To illustrate, let us consider a bipartite state described by the density operator $\rho_{\beta}$ belongs to a $3 \otimes 3$ dimensional system as \cite{Rashi_2025}
\begin{equation}
	\rho_{\beta}=\frac{1}{5+2\beta^2}\sum_{i=1}^{3}\ket{\epsilon_i}\bra{\epsilon_i},~~\frac{1}{\sqrt{2}}\leq\beta\leq1
	\label{betastate}
\end{equation}
where $\ket{\epsilon_i}=\ket{0i}-\beta\ket{i0},~i=1,2$ and $\ket{\epsilon_3}=\sum_{j=0}^{2}\ket{jj}$.\\
The state $\rho_{\beta}$ represent a NPT entangled state in the whole range i.e. $\frac{1}{\sqrt{2}}\leq\beta\leq1$ but we find that CCNR criterion detect $\rho_{\beta}$ as an entangled state when $\beta \in [\frac{1}{\sqrt{2}},1]$. However, we find that not a single state belong to the family of the NPT entangled state $\rho_{\beta}$ is detected by DV criterion. 
\subsubsection{Construction of LEW Operator Using DV Criterion}
To strengthen the DV criterion, we construct a linear entanglement witness operator using the DV criterion itself. The prescription for the construction of linear witness operator can be given in the following theorem.\\
\begin{theorem}
 For a bipartite state $\rho_{AB}$ lying in $d \otimes d$ dimensional system, a LEW operator can be constructed as 
\begin{equation}
	\begin{split}
		W_L^{DV}=\Big(\frac{d(d-1)}{2}-d^{2}(d+1)^2~k\Big)I_{d^{2}}
		-\frac{1}{\lambda_{\max}(\rho_{AB})}C^{can}
	\end{split}
	\label{WDV1}
\end{equation}
where the terms $k$, $C^{can}$ and $\lambda_{\max}(\rho_{AB})$ have their usual meaning and defined in the earlier sections.
\end{theorem}
\begin{proof}
(i) If $\rho_{AB}^{sep}$ represents an arbitrary separable state in $d\otimes d$ dimensional system then the expectation value of the linear operator $W_L^{DV}$ is given by  
\begin{equation}
\begin{split}
Tr(W_L^{DV}\rho_{AB}^{sep})=&\Big(\frac{d(d-1)}{2}-d^2(d+1)^2k\Big)\\
&-\frac{1}{\lambda_{\max}(\rho_{sep})}Tr(C^{can}\rho_{sep})
\end{split}
\label{WDV2}
\end{equation} 
Using the results given in (\ref{result-1}) and (\ref{traceA less norm A}), the expression of the expectation value given in (\ref{WDV2}) reduces to
\begin{equation}
	\begin{split}
		Tr(W_L^{DV}\rho_{AB}^{sep})\geq&\frac{d(d-1)}{2}-||C^{can}||_1-d^2(d+1)^2k\\
		\end{split}
	\label{WDV3}
\end{equation}
Using the inequality $k\leq0$ given in (\ref{det k constant}) and from the DV criterion given in (\ref{dv1}), we get 
\begin{equation}
	Tr(W_L^{DV}\rho_{AB}^{sep})\geq0, \forall~ \rho_{AB}^{sep}
	\label{dvwit1}
\end{equation}
(ii) It can be easily verified that there exist an entangled state $\rho_{AB}^{(1)}$ defined in (\ref{dstate}) for which the expectation value of the operator $W_L^{DV}$ is negative i.e. 
\begin{equation}
	\begin{split}
		Tr(W_L^{DV}\rho_{1})&=-0.152209\\
		&<0
	\end{split}
\label{dvwit2}
\end{equation}
Thus, the linear operator $W_L^{DV}$ satisfies both the properties of a witness operator and thus we may consider $W_L^{DV}$ as a LEW operator.
\end{proof}
It has already been discussed earlier that not a single state belong to the class of state $\rho_{\beta}$ detected by the DV criterion. But we find that the LEW operator $W_L^{DV}$ defined in this section detect a few states from a class of state described by the density operator $\rho_{\beta}$. The linear witness operator detect a state $\rho_{\beta}$ when $\beta \in [0.7308, 0.7889]$.  
\subsubsection{Construction of NLEW Operator Using the LEW Operator $W_L^{DV}$}
In this section, we will construct a non-linear entanglement witness operator denoted by $W_{NL}^{DV}$ using the linear witness operator  $W_{L}^{DV}$ with the motivation that $W_{NL}^{DV}$ perform better than $W_L^{DV}$. The prescription of the construction is given below:\\
\begin{theorem}
	If  $W_L^{DV}$ denote a linear entanglement witness operator defined in (\ref{WDV1}) and $\lambda_{max}(W_L^{DV})$ denote the maximum eigenvalue of $W_L^{DV}$ then the non-linear entanglement witness operator may be defined as  
\begin{equation}
	W_{NL}^{DV}=W_L^{DV}-\frac{1}{\lambda_{\max}(W_L^{DV})}(W_L^{DV})^2
	\label{WNLDV}
\end{equation}
\end{theorem}
\begin{proof}
The proof of Theorem 7 is similar to the proof of Theorem 6.
\end{proof}
If we again recalling the state $\rho_{\beta}$ defined in (\ref{betastate}) then we can find that the NLEW operator $W_{NL}^{DV}$ detect more states than LEW operator $W_{L}^{DV}$. The witness operator $W_{NL}^{DV}$ detect the state $\rho_{\beta}$ when the parameter $\beta \in [0.7308,0.8096]$.\\
Therefore, one may conclude in this section that although NLEW operator constructed directly from LEW operator and indirectly from the DV separability criterion but NLEW operator may serve as an efficient entanglement witness operator in comparison to both the LEW operator and the DV separability criterion.
\section{Construction Of NLEW Operator for  PPTES in $2 \otimes d_{2}$ $(d_{2}> 3)$ and $d_{1} \otimes d_{2}$  $(d_{1}> 2, d_{2}\geq 3)$ dimensional system and analyzing its efficiency}
In the previous section, we have constructed NLEW operator and studied their efficiencies but those constructed NLEW operators work for the detection of NPTES only and almost fails to detect PPTES. The statement that the NLEW operator $W_{NL}^{(1)}$ fails to detect any PPTES is true due to the following reasons: (i) The quantity $det(I_{d_{2}}+Tr_{A}(\rho_{AB}))-det(I_{d_{1}d_{2}}+\rho_{AB} < 0$ for all PPT state and (ii) we assume that the used LEW operator in the construction of $W_{NL}^{(1)}$ is a decomposable witness operator.\\
Instead, if the LEW is a non-decomposable witness operator then also we find that there may be a little chance that $W_{NL}^{(1)}$ detect the PPTES. The possible reason behind this is due to the structure of the operator $W_{NL}^{(1)}$.\\
In a similar manner, one can check that the NLEW operator $W_{NL}^{(2)}$ and $W_{NL}^{(3)}$ mainly take part in the detection of $NPTES$. Further, we note that the NLEW operator $W_{NL}^{CCNR}$ and $W_{NL}^{DC}$ also detects only NPTES. Thus, it is necessary to discuss the construction of NLEW operator for the detection of PPTES separately, which we would do in the next section.  
\subsection{Construction of NLEW operator for the detection of PPTES}
In this section, our aim is to construct a NLEW operator that may be used mainly to detect PPTES but they may detect NPTES also.\\
\begin{theorem}
 Let us consider a bipartite quantum state described by the density operator $\rho_{AB}$ in $d_{1} \otimes d_{2}$ dimensional system. If $W_L$ is a LEW operator  in $d_1\otimes d_2$ dimensional system  and $\underset{\rho_{AB}^{sep}}{max} (Tr(W_L^{2}\rho_{AB}^{sep}))$ gives the maximum value of $Tr(W_L^{2}\rho_{AB}^{sep})$ over all the separable states $\rho_{AB}^{sep}$ and $I_{d_1d_2}$ is an identity matrix of order $d_{1}d_{2}$ then the NLEW operator may be constructed as
\begin{equation}
\begin{split}
	W_{NL}^{(4)}&=\underset{\rho_{AB}^{sep}}{max} (Tr(W_L^{2}\rho_{AB}^{sep}))I_{d_{1}d_{2}}-W_{L}^2+\\
	&\Bigg(det(\rho_{AB})-\Big(\frac{det(Tr_{A}(\rho_{AB}))}{d_1}\Big)^{d_{1}}\Bigg)W_{L}
\end{split}
\label{NPTES AND PPTES WO}
\end{equation}
\end{theorem}
\begin{proof}
Let us consider an arbitrary separable state described by the density operator $\rho_{AB}^{sep}$ in $d_1\otimes d_2$ dimensional  system. Then the expectation of $W_{NL}^{(4)}$ with respect to $\rho_{AB}^{sep}$ is given by 
\begin{equation}
\begin{split}
	Tr(W_{NL}^{(4)}\rho_{AB}^{sep})&=\underset{\rho_{AB}^{sep}}{max} (Tr(W_L^{2}\rho_{AB}^{sep}))-Tr(W_{L}^2\rho_{AB}^{sep})+\\
	&\Bigg(det(\rho_{AB}^{sep})-\Big(\frac{det(Tr_{A}(\rho_{AB}^{sep}))}{d_1}\Big)^{d_{1}}\Bigg) \times \\& Tr(W_{L}\rho_{AB}^{sep})
\end{split}
\label{NPTES AND PPTES Tr}
\end{equation}
Let us define the two quantities $h_{1}$ and $h_{2}$ as
\begin{align}
h_1=	\underset{\rho_{AB}^{sep}}{max}(Tr(W_L^{2}\rho_{AB}^{sep}))-Tr(W_L^{2}\rho_{AB}^{sep})
\label{h1}
\end{align} 
\begin{align}
h_2=\Bigg(det(\rho_{AB}^{sep})-\Big(\frac{det(Tr_{A}(\rho_{AB}^{sep}))}{d_1}\Big)^{d_{1}}\Bigg)
\label{h2}
\end{align}
Using (\ref{h1}) and (\ref{h2}), the expression for $Tr(W_{NL}^{(4)}\rho_{AB}^{sep})$ given in (\ref{NPTES AND PPTES Tr}) may be re-expressed as
\begin{equation}
\begin{split}
	Tr(W_{NL}^{(4)}\rho_{AB}^{sep})&=h_{1} +h_{2}Tr(W_{L}\rho_{AB}^{sep})\\& \geq h_1+h_2\lambda_{min}(W_{L})
\end{split}
\label{NPTES AND PPTES Last equation}
\end{equation}
In the second line, we have used Result-1 given in (\ref{result-1}).\\
Using the inequality given in (\ref{det inequality from lin}), we find that $h_{2}\leq 0$. Since $W_{L}$ is an EW operator so it must have atleast one negative eigenvalue so $\lambda_{min}(W_{L})<0$. Also, we have $h_{1}\geq 0$. Therefore, the expectation value given in (\ref{NPTES AND PPTES Last equation}) must satisfy the following inequality
\begin{equation}
Tr(W_{NL}^{(4)}\rho_{AB}^{sep})\geq0,~~\forall~ \rho_{AB}^{sep}
\end{equation}
In the next step, let us consider 2$\otimes2$ NPTES  defined in (\ref{rho_entcons.}) as
\begin{align}
	\rho_{ent}=\begin{pmatrix}
		\frac{13}{30}&0&0&\frac{11}{30}\\
		0&\frac{1}{15}&0&0\\
		0&0&\frac{1}{15}&0\\
		\frac{11}{30}&0&0&\frac{13}{30}
	\end{pmatrix}
\label{rhoent}
\end{align}
If we consider a LEW operator given in (\ref{WLP WO}) then we find that $Tr(W_{L}^p\rho_{ent})>0, ~\forall~0<p\leq1$ i.e $W_{L}^p$ can not identify the state $\rho_{ent}$ given in (\ref{rhoent}). But we can construct a NLEW operator defined in (\ref{NPTES AND PPTES WO}) using the quantities given below:
\begin{equation}
	\begin{split}
		\underset{\rho_{AB}^{sep}}{\max}(Tr((W_L^p)^{2}\rho_{AB}^{sep}))&=\frac{1}{4}(2-2p+p^2)\\ \Bigg(det(\rho_{ent})-\Big(\frac{det(Tr_{A}(\rho_{ent}))}{2}\Big)^{2}\Bigg)&=-\frac{16619}{1080000}
	\end{split}
\end{equation}
Therefore, the expectation value of the NLEW operator $W_{NL}^{(4)}$ w.r.t the state $\rho_{ent}$ is given by
\begin{equation}
	Tr(W_{NL}^{(4)}\rho_{ent})<0, \text{for some $p$,}~ 0<p\leq1
\end{equation}
Hence, $W_{NL}^{(4)} $ is a well defined NLEW operator.
\end{proof}
\subsection{Efficiency of NLEW operator $W_{NL}$ over a linear witness operator}
\noindent In this section, we will give two examples of PPTES in $3\otimes 3$ dimensional system, which is detected by the NLEW operator defined in (\ref{NPTES AND PPTES WO}).\\
\textbf{Example-7a} Let us consider a two-qutrit
PPTES defined as \cite{Majumdar_2021}:
\begin{align}
\rho_{x}=\frac{1}{3(1+x+\frac{1}{x})}\begin{pmatrix}
	1&0&0&0&1&0&0&0&1\\
	0&x&0&0&0&0&0&0&0\\
	0&0&\frac{1}{x}&0&0&0&0&0&0\\
	0&0&0&\frac{1}{x}&0&0&0&0&0\\
	1&0&0&0&1&0&0&0&1\\
	0&0&0&0&0&x&0&0&0\\
	0&0&0&0&0&0&x&0&0\\
	0&0&0&0&0&0&0&\frac{1}{x}&0\\
	1&0&0&0&1&0&0&0&1
\end{pmatrix}
\label{rhox}
\end{align}
where $x$ is a positive real number.\\
Further, let us again consider a LEW operator defined in (\ref{WLCWO}) as
\begin{align}
	W_{L}^{c}=\frac{1}{33}\begin{pmatrix}
		1&0&0&0&-1&0&0&0&-1\\
		0&9&0&0&0&0&0&0&0\\
		0&0&1&0&0&0&0&0&0\\
		0&0&0&1&0&0&0&0&0\\
		-1&0&0&0&1&0&0&0&-1\\
		0&0&0&0&0&9&0&0&0\\
		0&0&0&0&0&0&9&0&0\\
		0&0&0&0&0&0&0&1&0\\
		-1&0&0&0&-1&0&0&0&1\\
	\end{pmatrix}
\label{wlc}
\end{align} 
 For the qutrit state (\ref{rhox}), we find that $W_{L}^c$ is not able to identify this state $\rho_{x}$ i.e $Tr(W_{L}^{c}\rho_{x})\geq0,~\forall x>0$.\\
 Using LEW operator (\ref{wlc}), let us calculate the following quantities that would be needed to construct NLEW operator defined in (\ref{NPTES AND PPTES WO}). 
\begin{equation}
\begin{split}
	\underset{\rho_{sep}}{\max}(Tr((W_L^c)^{2}\rho_{sep}))&=0.040155\\
	Tr((W_{L}^{c})^{2}\rho_{x})&=\frac{1+x+81x^2}{1089(1+x+x^2)}\\
	Tr(W_{L}^{c}\rho_{x})&=\frac{1 + x (-1 + 9 x)}{33 (1 + x + x^2)}\\
	\Bigg(det(\rho_{x})-\Big(\frac{det(Tr_{A}(\rho_{x}))}{3}\Big)^{3}\Bigg)&=-\frac{1}{531441}
\end{split}
\label{x state values}
\end{equation}
Using (\ref{NPTES AND PPTES WO}) and (\ref{x state values}), the expectation value of NLEW operator $W_{NL}^c$ with respect to the state $\rho_{x}$ is given by 
\begin{equation}
\begin{split}
	Tr(W_{NL}^c\rho_{x})&=\frac{-9329841 + 1574640\sqrt{3}+\frac{2(7085924 + 7085935 x)}{1 + x + x^2}}{192913083}\\
	&<0,~~~x\geq1.79
\end{split}
\end{equation}
Therefore, the NLEW operator $W_{NL}^{c}$ detect the state $\rho_{x}$ as an entangled state when $x\geq1.79$.\\ 
\textbf{Example-7b} Let us consider a  bipartite quantum state described by the density operator $\rho_\gamma$, which is defined as \cite{Horodecki_1999}:
\begin{align}
\rho_\gamma=	\frac{2}{7}\ket{\psi^+}\bra{\psi^+}+\frac{\gamma}{7}\sigma_{+}+\frac{5-\gamma}{7}\sigma_{-},~~2\leq\gamma\leq5
\end{align}
where $\ket{\psi^+}=\frac{1}{\sqrt{3}}(\ket{00}+\ket{11}+\ket{22}) $ and
\begin{align}
\begin{split}
	\sigma_{+}&=\frac{1}{3}(\ket{01}\bra{01}+\ket{12}\bra{12}+\ket{20}\bra{20})\\
	\sigma_{-}&=\frac{1}{3}(\ket{10}\bra{10}+\ket{21}\bra{21}+\ket{02}\bra{02})
\end{split}
\end{align}
For different values of the parameter $\gamma$, the state $\rho_\gamma$ may be represented as
\begin{equation*}
	\begin{cases}
	\text{Separable}~~ when~~2\leq\gamma\leq3\\
	\text{PPTES}~~ when~~3<\gamma\leq4\\
	\text{NPTES}~~ when~~4<\gamma\leq5
\end{cases}
\end{equation*} 	
One can easily calculate the following
\begin{align}
	\begin{split}
		Tr(W_{L}^{c}\rho_{\gamma})=\frac{3+8\gamma}{231}
		\geq 0,~2\leq\gamma\leq5
	\end{split}
\end{align}
The above inequality implies that $W_{L}^{c}$ cannot identify any member of the family of entangled state. Further, let us calculate the following quantities to get the expectation value of NLEW operator $W_{NL}^{c}$.
\begin{equation}
\begin{split}
	\underset{\rho_{sep}}{\max}(Tr((W_L^c)^{2}\rho_{sep}))&=0.040155\\
	Tr((W_{L}^{c})^{2}\rho_\gamma)&=\frac{7+80\gamma}{7623}\\
	Tr(W_{L}^{c}\rho_{\gamma})&=\frac{3+8\gamma}{231}\\
	\Bigg(det(\rho_{\gamma})-\Big(\frac{det(Tr_{A}(\rho_{\gamma}))}{3}\Big)^{3}\Bigg)&=-\frac{1}{531441}
\end{split}
\label{gamma state values}
\end{equation}
Using (\ref{gamma state values}), the expectation value of the NLEW operator $W_{NL}^{c}$ given in (\ref{NPTES AND PPTES WO}) with respect to the state $\rho_{\gamma}$ can be calculated as
\begin{equation}
\begin{split}
	Tr(W_{NL}^c\rho_{\gamma})&=\frac{33894093 + 11022480\sqrt{3}-14171848\gamma}{1350391581}\\
	&<0,~~3.74\leq\gamma\leq5
\end{split}
\end{equation}
Thus, NLEW operator $W_{NL}^c$ can detect all NPTES as well as many PPTES of the family described by the density operator $\rho_{\gamma}$.
\\
\section{Decomposition of NLEW operators in terms of the tensor product of local observables}
In this section we will be showing that how NLEW that we have constructed can be written as a tensor product of local observables. It is required to show that the constructed NLEW operator can be implemented in an experimental set up.
\subsection{For $2\otimes2$ dimensional system}
To proceed toward our aim, we recall a LEW operator $W_{L}^{p}$ defined in (\ref{WLP WO}) for the construction of NLEW operator defined in (\ref{WNL1}).\\ 
Let us start with the operator form of Pauli matrices, which can be expressed as
\begin{equation}
\begin{split}
	&\sigma_{x}=\ket{0}\bra{1}+\ket{1}\bra{0},~	\sigma_{y}=-\iota\ket{0}\bra{1}+i\ket{1}\bra{0}\\
	&\sigma_{z}=\ket{0}\bra{0}-\ket{1}\bra{1},~I_{2}=\ket{0}\bra{0}+\ket{1}\bra{1}
\end{split}
\label{pauli matrices}
\end{equation}
Expressing (\ref{pauli matrices}) in an another way, we get 
\begin{equation}
	\begin{split}
	&\ket{1}\bra{0}=\frac{\sigma_{x}-\iota\sigma_{y}}{2},~~\ket{0}\bra{1}=\frac{\sigma_{x}+\iota\sigma_{y}}{2}\\
	&\ket{0}\bra{0}=\frac{I_{2}+\sigma_{z}}{2},~~\ket{1}\bra{1}=\frac{I_{2}-\sigma_{z}}{2}
	\end{split}
	\label{ketbra values}
\end{equation}
The LEW operator $W_{L}^p$ may be expressed in the computational basis as 
\begin{equation}
	\begin{split}
		W_{L}^p=&\frac{1}{2}\Big[p\Big(\ket{0}\bra{0}\otimes \big(\ket{0}\bra{0}-\ket{1}\bra{1}\big)+\\
		&\ket{1}\bra{1}\otimes\big(\ket{1}\bra{1}-\ket{0}\bra{0}\big)\Big)\\
		&+\Big(\ket{0}\bra{0}\otimes\ket{1}\bra{1}+\ket{0}\bra{1}\otimes\ket{1}\bra{0}+\\
		&\ket{1}\bra{0}\otimes\ket{0}\bra{1}+\ket{1}\bra{1}\otimes\ket{0}\bra{0}\Big)\Big]
	\end{split}
	\label{WLPIN KETBRA2}
\end{equation}
In terms of Pauli operators, the LEW operator given in (\ref{WLPIN KETBRA2}) may be re-expressed as
\begin{equation}
	W_{L}^p=\frac{1}{4}\big[(2p-1)\sigma_{z}\otimes\sigma_{z}+\sigma_{x}\otimes\sigma_{x}+\sigma_{y}\otimes\sigma_{y}+I_{2}\otimes I_{2}\big]
	\label{WLP DECOMPOSITON}
\end{equation}
The operator $(W_{L}^p)^{2}$ may be expressed in terms of Pauli operators as
\begin{equation}
	\begin{split}
		(W_{L}^p)^2&=\frac{1}{4}\big((p-1)\sigma_{z}\otimes\sigma_{z}+(1+p^2-p)I_{2}\otimes I_{2}+\\
		&2(1-p)(\sigma_{x}\otimes\sigma_{x}+\sigma_{y}\otimes\sigma_{y})\big)
	\end{split}
	\label{WLP2 DECOMPOSITON}
\end{equation}
Therefore, using (\ref{WLP DECOMPOSITON}) and (\ref{WLP2 DECOMPOSITON}), the NLEW operator $W_{NL}^{(1)}$ defined in (\ref{WNL1}) can also be expressed in terms of the Pauli operators as
\begin{equation}
	\begin{split}
		W_{NL}^{(1)}=&\frac{37p-21}{80}\sigma_{z}\otimes\sigma_{z}+\frac{13-5p}{40}(\sigma_{x}\otimes\sigma_{x}+\\
		&\sigma_{y}\otimes\sigma_{y})+\bigg(\frac{5p^2-5p+21}{80}-\\
		&\big(det(I_{2}+Tr_{A}(\rho_{AB}))-det(I_{4}+\rho_{AB})\big)\bigg)I_{2}\otimes I_{2}
	\end{split}
\end{equation}
Since the NLEW operator $W_{NL}^{(1)}$ is expressed in terms of the tensor product of the local observables so it can be implemented in the laboratory.

\subsection{For $3\otimes3$ dimensional system}
In this case we will be considering the witness operator $W_{L}^c$ defined in (\ref{WLCWO}).\\
For $3\otimes3$ dimensional system, the operator form of Gell–Mann matrices in the computational basis can be expressed as
\begin{equation}
	\begin{split}
		&	\tau_{1}=\ket{0}\bra{1}+\ket{1}\bra{0},~~~~~\tau_{2}=\ket{0}\bra{2}+\ket{2}\bra{0}\\
		&\tau_{3}=\ket{1}\bra{2}+\ket{2}\bra{1},~~~~~\tau_{4}=-\iota\ket{0}\bra{1}+\iota\ket{1}\bra{0}\\
		&\tau_{5}=-\iota\ket{0}\bra{2}+\iota\ket{2}\bra{0},~\tau_{6}=-\iota\ket{1}\bra{2}+\iota\ket{2}\bra{1}\\
		&\tau_{7}=\ket{0}\bra{0}-\ket{1}\bra{1},~~~~~~\tau_{8}=\frac{1}{\sqrt{3}}\ket{0}\bra{0}+\ket{1}\bra{1}-2\ket{2}\bra{2}\\
		&I_3=\ket{0}\bra{0}+\ket{1}\bra{1}+\ket{2}\bra{2}
	\end{split}
	\label{gell mann1}
\end{equation}
Equation (\ref{gell mann1}) can also be re-expressed as
\begin{equation}
	\begin{split}
		&	\ket{1}\bra{0}=\frac{\tau_1-\iota\tau_4}{2},	\ket{0}\bra{1}=\frac{\tau_1+\iota\tau_4}{2},\\
		&\ket{2}\bra{1}=\frac{\tau_3-\iota\tau_6}{2},\ket{1}\bra{2}=\frac{\tau_3+\iota\tau_6}{2}\\
		&\ket{2}\bra{0}=\frac{\tau_2-\iota\tau_5}{2},\ket{0}\bra{2}=\frac{\tau_2+\iota\tau_5}{2}\\
		&\ket{0}\bra{0}=\frac{1}{6}(3\tau_7+\sqrt{3}\tau_8+2I_3)\\
		&\ket{1}\bra{1}=\frac{1}{6}(2I_3+\sqrt{3}\tau_8-3\tau_7)\\
		&\ket{2}\bra{2}=\frac{1}{3}(I-\sqrt{3}\tau_8)
	\end{split}
	\label{gell mann2}
\end{equation}
In terms of computational basis, the LEW operator $W_{L}^c$ can be written as:
\begin{equation}
	\begin{split}
		W_{L}^c&=\frac{1}{33}\big[\ket{0}\bra{0}\otimes(\ket{0}\bra{0}+\ket{2}\bra{2}+9\ket{1}\bra{1})+\\
		&\ket{1}\bra{1}\otimes(\ket{0}\bra{0}+\ket{1}\bra{1}+9\ket{2}\bra{2})+\\
		&\ket{2}\bra{2}\otimes(\ket{1}\bra{1}+\ket{2}\bra{2}+9\ket{0}\bra{0})-\\
		&\ket{0}\bra{1}\otimes\ket{0}\bra{1}-\ket{0}\bra{2}\otimes\ket{0}\bra{2}-\ket{1}\bra{0}\otimes\ket{1}\bra{0}-\\
		&\ket{1}\bra{2}\otimes\ket{1}\bra{2}-\ket{2}\bra{0}\otimes\ket{2}\bra{0}-\ket{2}\bra{1}\otimes\ket{2}\bra{1}\big]	
	\end{split}
	\label{wlc comptutational}
\end{equation}
Further, the LEW operator $W_{L}^c$ can be expressed in terms of Gell-Mann matrices as
\begin{equation}
	\begin{split}
		W_{L}^c&=\frac{1}{33}\bigg[\frac{11}{3}I_3\otimes I_3-\frac{1}{2}\sum_{i=1}^{3}\tau_i \otimes\tau_i+\frac{1}{2}\sum_{i=4}^{6}\tau_i \otimes\tau_i\\
		&-2\sum_{i=7}^{8}\tau_i \otimes\tau_i+2\sqrt{3}(\tau_7\otimes\tau_8-\tau_8\otimes\tau_7)\bigg]
	\end{split}
	\label{wlc gell mann1}
\end{equation}
Using (\ref{wlc gell mann1}) and the expression of $(W_{L}^c)^2$
in terms of Gell-Mann matrices, we can express the NLEW operator $W_{NL}^{(1)}$ in terms of Gell-Mann matrices as
\begin{equation}
	\begin{split}
		W_{NL}^{(1)}=&\frac{1}{33}\bigg[-\frac{2768}{1485}\sum_{i=7}^{8}\tau_i \otimes\tau_i+\frac{2683}{5940}\big(\sum_{i=4}^{6}\tau_i \otimes\tau_i-\\
	&	\sum_{i=1}^{3}\tau_i \otimes\tau_i\big)+\frac{30503(\sqrt{3})}{16335}(\tau_7\otimes\tau_8-\tau_8\otimes\tau_7)\bigg]+\\
		&\bigg[\frac{30253}{294030}-\big(det(I_{2}+Tr_{A}(\rho_{AB}))-\\
		&det(I_{4}+\rho_{AB})\big)\bigg]{I_3}\otimes{I_3}
	\end{split}
\end{equation}
Thus, the NLEW operator $W_{NL}^(1)$ can be expressed as tensor product of local observables in $3\otimes3$ dimensional system also hence we can conclude that our framework can be implemented in an experimental settings.
\section{Conclusion}
To summarize, we systematically addressed one of the fundamental problems in quantum information theory, which is, how to reliably detect entanglement in bipartite quantum systems. It is well known that constructing entanglement witnesses, mainly linear entanglement witnesses (LEWs) is an arduous task because of its dependence on prior partial knowledge of quantum states. More significantly, even after the construction of LEWs, these witnesses have limited detection ability and often fail to capture a wide range of entangled states within a given family. This limitation emphasizes that no single criterion can detect all the entangled states.\\
Motivated by this limitation, we developed a systematic and constructive framework to generate nonlinear entanglement witnesses (NLEWs) from a given LEW such that it can detect negative-partial-transpose entangled state (NPTES). Here, it may have a slight chance that this construction can also detect positive-partial-transpose entangled state (PPTES), if we assume LEW to be non-decomposable in spite of LEW actually being decomposable. The proposed method not only points out the limitations of linear witnesses, but also improves detection capability. This improvement can be explained geometrically in the following way: a LEW corresponds to a hyperplane that segregates some entangled states from separable states, whereas NLEW forms a curved surface \cite{Karimipour_2025}. Thus, NLEWs constructed from some given LEWs are able to identify a large number of entangled states that remain undetected by LEW's.\\
Furthermore, we showed how we can use a well established entanglement detection criteria's, namely, CCNR and DV criterion to construct LEW's. Then, we extended the LEW's, which are constructed from separability criterion, into NLEWs. In this linear and non linear theoretical foundation of CCNR, we found how efficiently NLEW operators are detecting the entangled states as compared to linear ones, which is however equivalent to CCNR criterion. But in the case of DV criterion, we observed that NLEW constructed obliquely from DV is more efficacious in contrast to DV criterion itself and LEW constructed from DV criterion.\\
Moreover, we proposed another framework of the non-linear improvement that proves to be more versatile and efficient in detecting the positive-partial-transpose entangled state (PPTES) which are hard to recognize with usual methods.\\
Adding to this,  we have shown that the constructed NLEW can be decomposed into the tensor product of local observables which establishes the experimental feasibility of NLEW.
Therefore, we can conclude that the non-linear improvement studied in this work supports the current progress that non-linear techniques can detect entanglement, which the standard linear methods cannot.

\appendix
\section{Detailed calculation of the expression $max_{\rho_{sep}}[Tr(W_L^{2}\rho_{sep})]$ for $2\otimes2$ and $3\otimes 3$ system}
	\subsection{For $2\otimes2$ system}
A general single-qubit quantum state may be represented as $\rho=\frac{1}{2}(I+\vec{r}\vec{\sigma})$ where $\vec{r}=(r_{1},r_{2},r_{3}),~\sum_{1}^{3}r_{i}^{2}\leq1$ and $\vec{\sigma}=(\sigma_{1},\sigma_{2},\sigma_{3})$ are the Pauli's matrices.\\
To start the calculation, let us consider a bipartite product state $\rho_{A}^{(j)} \otimes \rho_{B}^{(j)}$ for any $j=1,2,3,.....$, where
\begin{align}
	\begin{split}
		\rho_{A}^{(j)}=&\frac{1}{2}(I+r_{A_1}^{(j)}\sigma_{1}+r_{A_2}^{(j)}\sigma_{2}+r_{A_3}^{(j)}\sigma_{3}) \\
		\rho_{B}^{(j)}=&\frac{1}{2}(I+r_{B_1}^{(j)}\sigma_{1}+r_{B_2}^{(j)}\sigma_{2}+r_{B_3}^{(j)}\sigma_{3})\\ 
	\end{split}
\end{align}
In matrix form, $\rho_{A}^{(j)} \otimes \rho_{B}^{(j)}$ can be expressed as
	\begin{widetext}
	\begin{align}
		\rho_{A}^{(j)}\otimes\rho_{B}^{(j)}	=\begin{pmatrix}
			\frac{1}{4}( 1+r_{A_3}^{(j)})(1+r_{B_3}^{(j)})&   \frac{1}{4}( 1+r_{A_3}^{(j)})(r_{B_1}^{(j)}-\iota r_{B_2}^{(j)})& \frac{1}{4}(r_{A_1}^{(j)}-\iota r_{A_2}^{(j)})(1+r_{B_3}^{(j)})&\frac{1}{4} (r_{A_1}^{(j)}-\iota r_{A_2}^{(j)})(r_{B_1}^{(j)}-\iota r_{B_2}^{(j)})\\
			\\
			\frac{1}{4}( 1+r_{A_3}^{(j)})(r_{B_1}^{(j)}+\iota r_{B_2}^{(j)})&\frac{1}{4}( 1+r_{A_3}^{(j)})(1-r_{B_3}^{(j)})&\frac{1}{4} (r_{A_1}^{(j)}-\iota r_{A_2}^{(j)})(r_{B_1}^{(j)}+\iota r_{B_2}^{(j)})&\frac{1}{4}(r_{A_1}^{(j)}-\iota r_{A_2}^{(j)})(1-r_{B_3}^{(j)})\\
			\\
			\frac{1}{4}(r_{A_1}^{(j)}+\iota r_{A_2}^{(j)})(1+r_{B_3}^{(j)})&\frac{1}{4} (r_{A_1}^{(j)}+\iota r_{A_2}^{(j)})(r_{B_1}^{(j)}-\iota r_{B_2}^{(j)})&\frac{1}{4}( 1-r_{A_3}^{(j)})(1+r_{B_3}^{(j)})&\frac{1}{4}( 1-r_{A_3}^{(j)})(r_{B_1}^{(j)}-\iota r_{B_2}^{(j)})\\
			\\
			\frac{1}{4} (r_{A_1}^{(j)}+\iota r_{A_2}^{(j)})(r_{B_1}^{(j)}+\iota r_{B_2}^{(j)})&	\frac{1}{4}(r_{A_1}^{(j)}+\iota r_{A_2}^{(j)})(1-r_{B_3}^{(j)})&\frac{1}{4}( 1-r_{A_3}^{(j)})(r_{B_1}^{(j)}+\iota r_{B_2}^{(j)})&\frac{1}{4}( 1-r_{A_3}^{(j)})(1-r_{B_3}^{(j)})
		\end{pmatrix}
	\end{align}
\end{widetext}
Any separable state can be decomposed as
	\begin{equation}
		\rho_{sep}=\sum_{j}^{}p_{j}\rho_{A}^{(j)}\otimes\rho_{B}^{(j)}
		\label{DM}
	\end{equation} 
For any general linear witness operator $W_{L}$ defined in $2 \otimes 2$ system, the expectation value of $(W_{L})^{2}$ with respect to any separable state $\rho_{sep}$ can be given as 
	\begin{equation}
		\begin{split}
			Tr((W_L)^2\rho_{sep})=&p_1 Tr((W_L)^2\rho_{A}^{(1)}\otimes\rho_{B}^{(1)})+\\
			&~p_2 Tr((W_L)^2\rho_{A}^{(2)}\otimes\rho_{B}^{(2)})+\\
			&~p_3 Tr((W_L)^2\rho_{A}^{(3)}\otimes\rho_{B}^{(3)})+\cdots
		\end{split}
		\label{A1 max wlp}
	\end{equation}
Now, if the maximum is taken over all state parameters $r_{A_i}^{(j)},r_{B_i}^{(j)}$ where $i=1,2,3$ and $j=1,2,3,\cdots$ then we have\\
	\begin{equation}
		\begin{split}
			Tr((W_L)^2\rho_{sep})	&\leq p_1\underset{r_{A_i}^{(1)},r_{B_i}^{(1)}}{\max} Tr((W_L)^2\rho_{A}^{(1)}\otimes\rho_{B}^{(1)})+\\
			&~p_2 \underset{r_{A_i}^{(2)},r_{B_i}^{(2)}}{\max} Tr((W_L)^2\rho_{A}^{(2)}\otimes\rho_{B}^{(2)})+\\
			&~p_3 \underset{r_{A_i}^{(3)},r_{B_i}^{(3)}}{\max} Tr((W_L)^2\rho_{A}^{(3)}\otimes\rho_{B}^{(3)})+\cdots
			\label{A2 max WLP}
		\end{split}
	\end{equation}
To go little forward, let us consider a specific form of the linear witness operator and thus recalling the linear witness operator $W_L^p$ given in (\ref{WLP WO}). For our convenience, let us re-writing it again as 
	\begin{align}
		W_{L}^{p}=\frac{1}{2}\begin{pmatrix}
			p&0&0&0\\
			0&1-p&1&0\\
			0&1&1-p&0\\
			0&0&0&p
		\end{pmatrix}, \text{where}~~0<p\leq1
	\label{wlp1}
	\end{align}
Using the operator $W_{L}^{p}$, the inequality (\ref{A2 max WLP}) can be re-expressed as
\begin{equation}
	\begin{split}
		Tr((W_L^{p})^2\rho_{sep})	&\leq p_1\underset{r_{A_i}^{(1)},r_{B_i}^{(1)}}{\max} Tr((W_L^{p})^2\rho_{A}^{(1)}\otimes\rho_{B}^{(1)})+\\
		&~p_2 \underset{r_{A_i}^{(2)},r_{B_i}^{(2)}}{\max} Tr((W_L^{p})^2\rho_{A}^{(2)}\otimes\rho_{B}^{(2)})+\\
		&~p_3 \underset{r_{A_i}^{(3)},r_{B_i}^{(3)}}{\max} Tr((W_L^{p})^2\rho_{A}^{(3)}\otimes\rho_{B}^{(3)})+\cdots
		\label{A2 max WLP1}
	\end{split}
\end{equation}
Squaring the operator $W_{L}^{p}$ given in (\ref{wlp1}) and calculating the $j^{th}$ term of the inequality (\ref{A2 max WLP1}), we get
	\begin{align}
		\begin{split}
			\underset{r_{A_i}^{(j)},r_{B_i}^{(j)}}{\max} Tr((W_L^{p})^2\rho_{A}^{(j)}\otimes\rho_{B}^{(j)})=&\frac{1}{4}\bigg(1+(-1+p)p+\\
			&(1-p)\underset{r_{A_i}^{(j)},r_{B_i}^{(j)}}{\max}\sum_{i=1}^{2}r_{A_i}^{(j)}r_{B_i}^{(j)}\\
			&-(1-p)\underset{r_{A_i}^{(j)},r_{B_i}^{(j)}}{\min}r_{A_3}^{(j)}r_{B_3}^{(j)}\bigg)\\=&\frac{1}{4}(2-2p+p^2)
	\end{split}
	\label{A8}
	\end{align}
Since equation (\ref{A8}) is a valid expression for all $j$, so the inequality (\ref{A2 max WLP1}) reduces to   
\begin{equation}
		\begin{split}
			Tr((W_L^p)^2\rho_{sep})&\leq p_1\frac{1}{4}(2-2p+p^2)+p_2\frac{1}{4}(2-2p+p^2)\\
			&~+p_3\frac{1}{4}(2-2p+p^2)+\cdots\\
			&\leq \frac{1}{4}(2-2p+p^2)\sum_{j}^{}p_j\\
			&= \frac{1}{4}(2-2p+p^2)
		\end{split}
	\end{equation}
In the last step, we use the fact that $\sum_{j}^{}p_j=1$.\\
Therefore, for any arbitrary separable state $\rho_{sep}$, we have
	\begin{equation}
		\underset{\rho_{sep}}{\max}Tr((W_L^p)^2\rho_{sep})=\frac{1}{4}(2-2p+p^2)
	\end{equation}
\subsection{For $ 3\otimes3$ system}
A general single-qutrit quantum state can be represented as $\rho=\frac{1}{3}(I+\vec{b}.\vec{\lambda})$, $\text{where}~ \vec{b}=(b_1,b_2,b_3,b_4,b_5,b_6,b_7,b_8)$ may be chosen in such a way that $\sum_{i=1}^{8}b_i^2\leq1$ and $\vec{\lambda}=(\lambda_{1},\lambda_{2},\lambda_{3},\lambda_{4},\lambda_{5},\lambda_{6},\lambda_{7},\lambda_{8})$, where $\lambda_i's$ denote the Gell Mann matrices.\\
Therefore, a two-qutrit in product state can be given as $\rho_{A}^{(j)}\otimes \rho_{B}^{(j)}$ for any $j=1,2,3,\cdots$, where 
	\begin{align}
		\begin{split}
			\rho_{A}^{(j)}=&\frac{1}{3}(I+b_{A_1}^{(j)}\lambda_{1}+b_{A_2}^{(j)}\lambda_{2}+b_{A_3}^{(j)}\lambda_{3}+b_{A_4}^{(j)}\lambda_{4}+\\
			&b_{A_5}^{(j)}\lambda_{5}+b_{A_6}^{(j)}\lambda_{6}+b_{A_7}^{(j)}\lambda_{7}+b_{A_8}^{(j)}\lambda_{8}) \\
			\rho_{B}^{(j)}=&\frac{1}{3}(I+b_{B_1}^{(j)}\lambda_{1}+b_{B_2}^{(j)}\lambda_{2}+b_{B_3}^{(j)}\lambda_{3}+b_{B_4}^{(j)}\lambda_{4}+\\
			&b_{B_5}^{(j)}\lambda_{5}+b_{B_6}^{(j)}\lambda_{6}+b_{B_7}^{(j)}\lambda_{7}+b_{B_8}^{(j)}\lambda_{8})\\ 
		\end{split}
	\end{align}
For any general linear witness operator $W_{L}$ defined in $3 \otimes 3$ system, the expectation value of $(W_{L})^{2}$ with respect to any separable state $\rho_{sep}$ in $3 \otimes 3$ system can be expressed as 
	\begin{equation}
		\begin{split}
			Tr((W_L)^2\rho_{sep})	&\leq p_1\underset{b_{A_i}^{(1)},b_{B_i}^{(1)}}{\max} Tr((W_L)^2\rho_{A}^{(1)}\otimes\rho_{B}^{(1)})+\\
			&~p_2 \underset{b_{A_i}^{(2)},b_{B_i}^{(2)}}{\max} Tr((W_L)^2\rho_{A}^{(2)}\otimes\rho_{B}^{(2)})+\\
			&~p_3 \underset{b_{A_i}^{(3)},b_{B_i}^{(3)}}{\max} Tr((W_L)^2\rho_{A}^{(3)}\otimes\rho_{B}^{(3)})+\cdots
			\label{B2 max WLP}
		\end{split}
	\end{equation}
The maximum is taken over all the state parameters $(b_{A_i}^{(j)},b_{B_i}^{(j)})$ for all $i=1,2\cdots8$ and $j=1,2,3,\cdots$.\\
Recalling the linear witness operator $W_L^{c}$ and re-writing it as
	\begin{align}
		W_{L}^{c}=\frac{1}{33}\begin{pmatrix}
			1&0&0&0&-1&0&0&0&-1\\
			0&9&0&0&0&0&0&0&0\\
			0&0&1&0&0&0&0&0&0\\
			0&0&0&1&0&0&0&0&0\\
			-1&0&0&0&1&0&0&0&-1\\
			0&0&0&0&0&9&0&0&0\\
			0&0&0&0&0&0&9&0&0\\
			0&0&0&0&0&0&0&1&0\\
			-1&0&0&0&-1&0&0&0&1\\
		\end{pmatrix}
	\label{A13}
	\end{align}
Using the linear witness operator $W_{L}^{c}$ in (\ref{B2 max WLP}), we get
\begin{equation}
	\begin{split}
		Tr((W_L^{c})^2\rho_{sep})	&\leq p_1\underset{b_{A_i}^{(1)},b_{B_i}^{(1)}}{\max} Tr((W_L^{c})^2\rho_{A}^{(1)}\otimes\rho_{B}^{(1)})+\\
		&~p_2 \underset{b_{A_i}^{(2)},b_{B_i}^{(2)}}{\max} Tr((W_L^{c})^2\rho_{A}^{(2)}\otimes\rho_{B}^{(2)})+\\
		&~p_3 \underset{b_{A_i}^{(3)},b_{B_i}^{(3)}}{\max} Tr((W_L^{c})^2\rho_{A}^{(3)}\otimes\rho_{B}^{(3)})+\cdots
		\label{B2 max WLP1}
	\end{split}
\end{equation}
The $j^{th}$ term of the inequality (\ref{B2 max WLP1}) can be calculated as
\begin{equation}
	\underset{b_{A_i}^{(j)},b_{B_i}^{(j)}}{\max}(Tr((W_{L}^{c})^{2}\rho_{A}^{(j)}\otimes\rho_{B}^{(j)}))=0.0401555
	\label{max WLC}
\end{equation}
Thus, for any arbitrary separable state $\rho_{sep}$, we have
\begin{equation}
	\underset{\rho_{sep}}{\max}	Tr((W_L^c)^2\rho_{sep})=0.0401555
\end{equation}
	

\begin{thebibliography}{9}
	\bibitem{EPR_1935} A. Einstein, B. Podolsky and N. Rosen, Phys. Rev. \textbf{47}, 777 (1935).
	
	\bibitem{Schrodinger_1935} E. Schrodinger, Naturewissenschaften, \textbf{23}, 807 (1935).
	
	\bibitem{Caltech} Caltech Science Exchange. (n.d.). What is entanglement and why is it important? California Institute of Technology.
	
	
	\bibitem{Wootters_1993} C. H. Bennett, G. Brassard, C. Crepeau, R. Jozsa, A. Peres, and	W. K. Wootters, Phys. Rev. Lett. \textbf{70} , 1895 (1993).
	
	\bibitem{X-Mhu_2023} X-Min Hu, Y. Guo, B-H Liu, C-F Li, G-C Guo, Nat. Rev. Phys. \textbf{5}, 339 (2023).
	
	\bibitem{Bennett_2001} C. H. Bennett, D. P. DiVincenzo , P . W. Shor, J . A .   Smolin, B . M. Terhal, and
	W . K. Wooters, Phys. Rev. Lett. \textbf{87}, 077902 (2001).
	
	\bibitem{pati_2000} A. K. Pati, Phys. Rev. A. \textbf{63}, 014302 (2000).
		
	\bibitem{Nielsen_2000} M. A. Nielsen and I. L. Chuang, Quantum Computation and Quantum Information (CambridgeUniversity,Cambridge,
	England,2000).
	
	\bibitem{bennett_1992} C. H. Bennett, and S. J. Wiesner, Phys. Rev. Lett. \textbf{69}, 2881 (1992).
	
	\bibitem{harrow_2004} A. Harrow, P. Hayden, and D. Leung, Phys. Rev. Lett. \textbf{92}, 187901 (2004).
	
	\bibitem{ekert_1991} A. K. Ekert, Phys. Rev. Lett. \textbf{67}, 661 (1991).
	
	\bibitem{brun_2006} T. Brun, I. Devetak, and M-H Hsieh, Science \textbf{314}, 436 (2006).
	
	\bibitem{bravyi_2025} S. Bravyi, D. Lee, Z. Li, and B. Yoshida, Phys. Rev. Lett. \textbf{134}, 210602 (2025).
	
	\bibitem{ekert_1998} A. K. Ekert, and R. Jozsa, Philos. Trans. A Math. Phys. Eng. Sci. \textbf{356}, 1769 (1998).
	
	\bibitem{Horodecki_2006} P. Horodecki, and R. Augusiak, Quant. Inf. Proc. \textbf{199}, 19 (2006).
	
	\bibitem{Adhikari_2018} S. Adhikari,  Phys. Rev. A. \textbf{97}, 042344(2018).
	
	\bibitem{peres_1996} A. Peres, Phys. Rev. Lett. \textbf{77}, 1413 (1996).
	
	\bibitem{Horodecki_1996} M. Horodecki, P. Horodecki, and R.Horodecki, Phys. Lett. A. \textbf{223}, 1 (1996).
	
	\bibitem{rohira_2021} R. Rohira, S. Sanduja, and S. Adhikari, Quant. Inf. Proc. \textbf{374}, 20 (2021).
	
	\bibitem{scala_2024} G. Scala, A. Bera, G. Sarbicki, and D. Chru\'{s}ci\'{n}ski, J. Phys. A: Math. Theor. \textbf{57}, 195301 (2024).
	
	\bibitem{Horodecki_1997}  P. Horodecki, Phys. Lett. A. \textbf{232}, 333 (1997).
	
	
	\bibitem{R Kumar_2026} R. Kumar, and S. Adhikari,  Phys. Lett. A. \textbf{567}, 131195 (2026).
	
	\bibitem{Horodecki_1999} M. Horodecki, and P. Horodecki, Phys. Rev. A. \textbf{59}, 4206(1999).
	
	\bibitem{Nielsen_2001} M. A. Nielsen, and J. Kempe, Phys. Rev. Lett. \textbf{86}, 5184 (2001). 
	
	\bibitem{Rudolph_2003} O. Rudolph, Phys. Rev. A. \textbf{67}, 032312 (2003).
	
	\bibitem{Chen_2002} K. Chen, L. Wu, Quant. Inf and Comp.\textbf{ 3}, 193 (2003).

	
	\bibitem{Terhal_2001} B. M. Terhal, Linear Algebra Appl. \textbf{323}, 61 (2001).
	
	\bibitem{Horodecki_2009} R. Horodecki, P. Horodecki, M. Horodecki, and K. Horodecki, Rev. Mod. Phys. \textbf{81}, 865 (2009).
	
	\bibitem{Vogel_2009} J. Sperling, and W. Vogel, Phys. Rev. A. \textbf{79}, 022318 (2009).
	
	\bibitem{Gurvits_2004} L. Gurvits, J. Comput. Syst. Sci. \textbf{69}, 448 (2004).
	
	\bibitem{Doherty_2004} A. C. Doherty, P. A. Parrilo, and F. M. Spedalieri, Phys. Rev. A. \textbf{69}, 022308 (2004).
	
	\bibitem{Hou_2010} J. Hou, and Y. Guo, Phys. Rev. A. \textbf{82}, 052301 (2010).
	
	\bibitem{Chruscinski_2009} D. Chru\'{s}ci\'{n}ski, J. Pytel, and G. Sarbicki, Phys. Rev. A. \textbf{80}, 062314 (2009).
	
	\bibitem{Duan_2000} L. Duan, G. Giedke, J.I. Cirac, and P. Zoller, Phys. Rev. Lett. \textbf{84}, 2722 (2000).
	
	\bibitem{Uffink_2002} J. Uffink, Phys. Rev. Lett. \textbf{88}, 230406 (2002).
	
	\bibitem{Hoffmann_2003}  H. F. Hofmann, and S. Takeuchi, Phys. Rev. A. \textbf{68}, 032103 (2003).
	
	\bibitem{Guhne_2006}  O.  G\"{u}hne, and N.  L\"{u}tkenhaus, Phys. Rev. Lett. \textbf{96}, 170502 (2006).
	
	\bibitem{Guhne_2004}  O.  G\"{u}hne, Phys. Rev. Lett. \textbf{92}, 117903 (2004).

	\bibitem{Karimipour_2025} A. Tangestaninejad, and V. Karimipour, 	arXiv:2506.23262 [quant-ph], (2025).
	
	\bibitem{Lasser_1995} J. B. Lasserre, IEEE Trans. on Automatic Control. \textbf{40}, 1500 (1995). 
	
	\bibitem{Lin_2016} M. Lin, Czech. Math. J. \textbf{66}, 737 (2016).
	
	\bibitem{Zhang_2019} P. Zhang, Lin. Alg. and its Appl. \textbf{576}, 258 (2019).
	
	\bibitem{Meenakshi_1999} A. R. Meenkshi, and C. Rajian,  Lin. Alg. and its Appl. \textbf{295}, 295 (1999).
	
	\bibitem{Zhan_2002} X. Zhan, "Matrix Inequalities" (Lecture notes in Mathematics), springer, Berlin. \textbf{ 1790} (2002). 
	
	\bibitem{Guhne_2009} O. G\"{u}hne, and G.T\'{o}th, Phys. Rep. \textbf{474}, 1 (2009). 
	
	\bibitem{Milne_2015} A. Milne, D. Jennings, and T. Rudolph, Phys. Rev. A, \textbf{92}, 012311 (2015).
	
	\bibitem{Chi_2003} D. P. Chi, and S. Li, J. Phys. A \textbf{36}, 11503 (2003).

	\bibitem{Sen_2023} K. Sen, C. Srivastava, and U. Sen, J. Phys. A: Math. Theor. \textbf{56}, 315301 (2023).
	
	\bibitem{Bertlmann_2009} R. A. Bertlmann, and P. Krammer, Ann. Phys. \textbf{324}, 1388 (2009).
	
		\bibitem{Hiroshima_2000} S. Ishizaka, and T. Hiroshima, Phys. Rev. A. \textbf{62}, 022310 (2000).
	
		\bibitem{Wang_2010} M-J Zhao, Z-G Li, S-M Fei, and Z-X Wang, J. Phys. A: Math. Theor. \textbf{43}, 275203 (2010).
		
		
	\bibitem{Wudarski_2011} D. Chru\'{s}ci\'{n}ski, and F.A. Wudarski, Open Sys. Information Dyn. \textbf{18}, 387 (2011).
	
\bibitem{Chruscinki_2020} G. Sarbicki, G Scala, and D Chruscinski, Phys. Rev. Lett., \textbf{101},	012341 (2020).

\bibitem{Rudolph_2005} O. Rudolph, Quant. Info. Proc. 4, \textbf{219} (2005).


		\bibitem{Vicenta_2007} J. D. Vicente, Quant. Inf. Comput.\textbf{ 7}, 624 (2007).
		
		\bibitem{Rashi_2025} R. Jain, and S. Adhikari, Phys. Scr., \textbf{100}, 125102 (2025).
		
		\bibitem{Majumdar_2021}  B. Bhattacharya, S. Goswami, R. Mundra, N. Ganguly, I. Chakrabarty, S. Bhattacharya, and A. S. Majumdar, J. Phys. Commun. \textbf{5}, 065008 (2021).
	
	
	
%
%
%
	
	
%
%
%
%
%
%
%
%
%
%
%
%
%
%
%
%
%
%
	

	
\end{thebibliography}
\end{document}